\documentclass[10pt,journal,twoside]{IEEEtran}
\IEEEoverridecommandlockouts
\usepackage[noadjust]{cite}

\usepackage{amsmath,amssymb,amsfonts}
\usepackage{algpseudocode}
\usepackage{algorithm}
\usepackage{amsthm}
\usepackage{enumitem}
\usepackage{graphicx}
\usepackage{textcomp}
\usepackage{xcolor}
\usepackage[colorlinks,citecolor=blue,linkcolor=blue,urlcolor=magenta]{hyperref}
\usepackage{booktabs}
\usepackage{cases}
\allowdisplaybreaks[4]
\usepackage{bbm,bm}
\usepackage{siunitx}
\usepackage{upgreek}
\usepackage{mathrsfs}
\usepackage{stfloats}
\usepackage{balance}
\usepackage{orcidlink}
\usepackage[caption=false, font=footnotesize]{subfig}
\usepackage{stfloats}
\usepackage{float}
\usepackage{comment}
\usepackage[mathscr]{eucal}

\theoremstyle{definition}
\newtheorem{theorem}{Theorem} 
\newtheorem{proposition}{Proposition} 
\newtheorem{definition}{Definition}
\newtheorem{remark}{Remark}

\renewcommand{\qed}{\hfill\IEEEQED}

\newcommand{\pd}[2]{\frac{\partial #1}{\partial #2}}

\newcommand{\herm}{\mathsf{H}}
\newcommand{\trans}{{\mathsf{T}}}

\newcommand{\Tr}{{\mathsf{tr}}}
\DeclareMathOperator{\BTr}{\overline{\bm{\mathsf{BTR}}}}
\DeclareMathOperator{\Mdiag}{\overline{\bm{\mathsf{DIAG}}}}
\DeclareMathOperator{\vdiag}{\overline{\bm{\mathsf{diag}}}}
\DeclareMathOperator{\Bodot}{\bm{\odot_{\mathsf{B}}}}
\DeclareMathOperator{\vect}{\overline{\bm{\mathsf{vec}}}}
\DeclareMathOperator{\Rank}{\mathsf{rank}}

\DeclareMathOperator{\st}{\mathrm{s.t.}}
\DeclareMathOperator{\sinc}{\mathrm{sinc}}

\newcommand{\boldtheta}{{\bm \theta}}
\newcommand{\boldtau}{{\bm \tau}}
\newcommand{\boldnu}{{\bm \nu}}
\newcommand{\boldLambda}{{\bm\Lambda}}
\newcommand{\boldXi}{{\bm\Xi}}
\newcommand{\boldDelta}{\pmb{\Delta}}
\newcommand{\jj}{\jmath}
\newcommand{\e}{\mathbf{e}}
\newcommand{\gr}{\g_{\mathsf r}}
\newcommand{\gi}{\g_{\mathsf i}}
\newcommand{\I}{\mathbf{I}}
\newcommand{\fc}{f_{\mathsf c}}
\newcommand{\gtf}{g_{\mathsf{tf}}}
\newcommand{\gdd}{g_{\mathsf{dd}}}
\newcommand{\noisevar}{\sigma^2}
\newcommand{\Ts}{T_{\mathsf s}}
\newcommand{\Tsym}{T_{\mathsf{sym}}}

\newcommand{\Mcp}{M_{\mathsf{pfx}}}
\newcommand{\SNR}{\varpi_{\mathsf{sym}}}
\renewcommand{\SNR}{\varpi}

\newcommand{\g}{\bm{g}}
\newcommand{\A}{\mathbf{A}}
\newcommand{\U}{\mathbf{U}}
\newcommand{\Rs}{\mathbf{\Phi}_{\s}}
\newcommand{\Rx}{\mathbf{\Phi}_{\x}}

\newcommand{\G}{\bm{\mathcal G}}
\newcommand{\NU}{\bm{\mathcal V}}
\newcommand{\TAU}{\bm{\mathcal T}_{\!\!\mathsf{ps}}}

\newcommand{\F}{\mathbf{F}}

\renewcommand{\S}{\underline{\mathbf{S}}}
\newcommand{\Scp}{\S_{\mathsf{pfx}}}
\newcommand{\Y}{\underline{\mathbf{Y}}}
\newcommand{\R}{\underline{\mathbf{R}}}
\newcommand{\Rcp}{\R_{\mathsf{pfx}}}
\newcommand{\X}{\mathbf{X}}
\newcommand{\Z}{\underline{\mathbf{Z}}}
\newcommand{\D}{\mathbf{D}}
\newcommand{\V}{\bm{V}}
\newcommand{\T}{\bm{T}}

\newcommand{\y}{\mathbf{y}}
\newcommand{\s}{\mathbf{s}}
\newcommand{\x}{\mathbf{x}}
\newcommand{\xw}{\x_{\mathsf{w}}}
\newcommand{\z}{\mathbf{z}}
\renewcommand{\r}{\mathbf{r}}

\newcommand{\FIM}{{\mathbf{J}}}
\renewcommand{\H}{{\mathbf{H}}}
\newcommand{\J}{\bm{\mathscr{J}}}
\newcommand{\CRB}{\mathbf{C}}

\newcommand{\OFDM}{\mathsf{OFDM}}
\newcommand{\AFDM}{\mathsf{AFDM}}
\newcommand{\ODDM}{\mathsf{ODDM}}
\newcommand{\OTFS}{\mathsf{OTFS}}
\newcommand{\OCDM}{\mathsf{OCDM}}

\newcommand{\sw}{\s_{\mathsf{w}}}

\newcommand{\Uw}{\U_{\mathsf{w}}}
\newcommand{\UOFDM}{\U_{\OFDM}}
\newcommand{\UOTFS}{\U_{\OTFS}}
\newcommand{\UAFDM}{\U_{\AFDM}}
\newcommand{\UOCDM}{\U_{\OCDM}}
\newcommand{\UODDM}{\U_{\ODDM}}

\newcommand{\Rsw}{\mathbf{\Phi}_{\s,\mathsf{w}}}

\newcommand{\Rxw}{\mathbf{\Phi}_{\x,\mathsf{w}}}
\newcommand{\pt}{p_{\mathsf{t}}}
\newcommand{\boldrho}{\bm{\varrho}}

\begin{document}
\title{On A Unified Cram\'er-Rao Bound\\ Framework for Joint Delay-Doppler\\ Estimation with Multi-Carrier Waveforms}
\author{
\IEEEauthorblockN{Zi-Jie Wang$^{\orcidlink{0009-0009-6702-1811}}$, Xudong Wang$^{\orcidlink{0000-0002-1353-1420}}$,~\IEEEmembership{Fellow,~IEEE}, and Giuseppe Caire$^{\orcidlink{0000-0002-7749-1333}}$,~\IEEEmembership{Fellow,~IEEE}}

\thanks{Zi-Jie Wang is with the Department of Electrical and Computer Engineering, National University of Singapore, Singapore 117583, and also with the Global College, Shanghai Jiao Tong University, Shanghai 200240, China (e-mail: zijie.wang@u.nus.edu).}
\thanks{Xudong Wang is with the Internet of Things Thrust, The Hong Kong University of Science and Technology (Guangzhou), Guangzhou, Guangdong 511455, China (e-mail: wxudong@ieee.org).}
\thanks{Giuseppe Caire is with the Faculty of Electrical Engineering and Computer Science, Technical University of Berlin, Berlin 10587, Germany (email: caire@tu-berlin.de).}
}

\markboth{IEEE Transactions on Signal Processing}{wang \lowercase{et al.}: CRB for joint delay-Doppler estimation with multi-carrier waveforms}
\maketitle

\begin{abstract}
Multi-carrier waveforms (e.g., OFDM, OTFS, ODDM, OCDM, AFDM, etc.) promise potentials for sensing in dynamic channels. Characterizing their fundamental estimation-theoretic limits is a key focus in signal processing community. In this paper, a unified, analytically tractable Cram\'{e}r-Rao bound (CRB) framework is established to benchmark joint delay-Doppler estimation performance with diverse multi-carrier waveforms. Under generic sensing setups, it accounts for effects of multi-target, unknown channel gains, channel dynamics, pulse shaping, etc., on the estimation process. Analysis reveals that CRBs are determined by the power profile of the transmitted \emph{time-frequency domain symbols}, 
rendering the CRB performance of different waveforms identical as long as their time-frequency symbol-level power distributions coincide. Hence, different waveforms yield the same theoretical limits on delay-Doppler CRBs; CRB gains may emerge when certain waveform-specific modulation reshapes the power distribution of some particular \emph{data} realization to a time-frequency \emph{symbol}-level power profile that is better aligned with the channel- and pulse-shaping-induced Fisher information matrix structure. To investigate fundamental performance limits of delay-Doppler estimation, the Pareto front of the delay-Doppler CRB region is identified by optimizing CRBs with respect to the power profiles. Corresponding solution naturally adapts to the delay-Doppler preference, pulse shaping spectrum, and delay-Doppler structure in channels. It provides a principled  waveform design guideline for optimal delay-Doppler estimation.
Finally, numerical results reveal i) the influence of various system parameters on the sensing performance, ii) the tightness of the derived bounds to practical estimators, and iii) the remarkable fact that some well-established CRB expressions (i.e., under uniform power among time-frequency symbols) can be far from Pareto-optimal CRBs, outlining the sub-optimality of such allocation strategies.
\end{abstract}

\section{Introduction}

\IEEEPARstart{W}{ith} the evolution toward the next generation of wireless systems, high-dynamic channels, e.g., in ultra-high mobility scenarios, non-stationary environments, etc., have become an intense subject of investigation \cite{Overview_waveform_OTFS,DD_waveform,Overview_waveform_AFDM,AFDM_communication,ODDM_overall,OCDM_overall,Multicarrier_survey}. These challenging conditions arise from fast speed, massive antenna arrays, mmWave/THz \cite{TVT_BA} band characteristics, and increasingly complex RF front-ends. To tackle them effectively, many advanced multi-carrier waveforms have emerged as alternatives to orthogonal frequency division multiplexing (OFDM), including delay-Doppler (DD) domain designs like orthogonal time-frequency space (OTFS) \cite{Overview_waveform_OTFS,OTFS_ISAC,entropy_OTFS}, orthogonal delay-Doppler division multiplexing (ODDM) \cite{ODDM_overall}, and chirp-based schemes such as orthogonal chirp-division multiplexing (OCDM) \cite{OCDM_overall,OCDM_original} and affine frequency division multiplexing (AFDM) \cite{AFDM_communication,Overview_waveform_AFDM}. These waveforms embrace complex channel dynamics rather than trying to suppress them, delivering superior robustness, lower processing complexity, and more efficient resource utilization \cite{Overview_waveform_OTFS,OCDM_overall,ODDM_overall,AFDM_communication}. Thanks to recent advances in integrated sensing and communication (ISAC) \cite{FanLiu6G,wang2025device,TVT_BA,wang_2026_TIT}, these waveforms are potentially useful for (radar) sensing as well \cite{Multicarrier_survey,braun2014ofdm}, holding great promise for accurate DD estimation in future high-mobility channels \cite{Overview_waveform_OTFS}.

Unveiling and understanding fundamental estimation-theoretic performance limits play a critical role in guiding system design and optimization, and serve as essential benchmarks for practical sensing algorithms \cite{wang2025device,TWC_Hua,wang_2026_TIT}. However, the majority of existing works primarily focus on signaling design and algorithm optimization, with sensing performance predominantly evaluated through numerical simulations (e.g., \cite{OFDM_AMT_TWC,entropy_OTFS,OTFS_ISAC,OTFS_ISAC_MIMO}), while only a limited number of studies conduct estimation-theoretic analysis.

Among the considered waveforms, OFDM is one of the few cases for which sensing-oriented theoretical analysis has been conducted in a relatively systematic manner. A comprehensive treatment of OFDM radar is provided in \cite{braun2014ofdm}, where the Cram\'{e}r-Rao bound (CRB) for joint range-velocity (equivalently, DD) estimation is characterized. In \cite{wang2025device}, closed-form Fisher information matrix (FIM) expressions are derived for joint range–angle estimation using OFDM-based ISAC signals, enabling the theoretical evaluation of multi-target localization performance in distributed MIMO radar systems. These studies establish an important analytical foundation for OFDM-based radar/ISAC implementations.

For OTFS, the Bayesian CRB under random parameter models is derived in \cite{OTFS_CRB_uplink}, and is subsequently employed for pilot design. The DD CRBs of OTFS and OFDM are derived and compared in \cite{OFDM+OTFS,OFDM+OTFS2}, where the resulting bounds are used to assess the performance of (almost) maximum-likelihood DD estimators. For AFDM, the CRB is derived in \cite{AFDM_ISAC} to examine the impact of the chirp-rate parameter on DD estimation performance. In \cite{cao2025agileaffinefrequencydivision}, the CRB is further expressed as a function of chirp parameters for Agile-AFDM, and these parameters are further optimized accordingly. Joint angle-DD CRBs for AFDM are developed in \cite{AFDM_ISAC_MIMO}, accounting for both near- and far-field targets. For ODDM, the CRB for joint angle-DD estimation is derived in \cite{ODDM_ISAC}, together with CRB-minimizing beamforming matrix designs for THz ISAC systems. In \cite{ODDM+FMCW}, a combined ODDM and frequency modulated continuous wave is proposed as a promising ISAC waveform, and its CRB performance on DD estimation is analyzed. Under a distributed MIMO radar setting, a generic joint position-velocity CRB is formulated in \cite{xia2025jointlocationvelocityestimation}, with numerical evaluations provided for waveforms such as OFDM and OCDM.

Although investigations of the estimation-theoretic performance of these waveforms have been conducted to different extents, several critical research gaps remain as follows.
\begin{enumerate}
	\item For certain waveforms, such as ODDM and OCDM, estimation-theoretic analysis remains relatively scarce, and systematic CRB characterizations are still missing.
	\item The CRB formulations reported in the literature are typically tailored to specific signal models, system assumptions, and parameterizations. For example, all the above-mentioned analyses are intrinsically waveform-specific, and a number of results are derived under simplified setups, such as single-target scenarios \cite{AFDM_ISAC,OFDM+OTFS,cao2025agileaffinefrequencydivision,ODDM_ISAC} and/or isotropic transmission models \cite{AFDM_ISAC,OFDM+OTFS}. 
	\item Most CRB results are primarily developed in forms of semi-analytical, numerical-computation-oriented expressions (e.g., involving deriving partial derivatives \cite{OFDM+OTFS,OFDM+OTFS2,ODDM+FMCW} and/or element-wise calculation of the FIM \cite{OTFS_CRB_uplink,cao2025agileaffinefrequencydivision,AFDM_ISAC_MIMO,ODDM_ISAC,ODDM+FMCW}). These expressions are used i) for numerical evaluation purposes to assess performance of certain algorithms \cite{OFDM+OTFS2,ODDM_ISAC,AFDM_ISAC_MIMO}; and/or ii) as sensing metrics for parameter or waveform optimization \cite{OTFS_CRB_uplink,ODDM_ISAC,cao2025agileaffinefrequencydivision}. In contrast, comparatively little attention has been devoted to extracting intrinsic insights from the CRBs/FIMs themselves, as done in \cite{xia2025jointlocationvelocityestimation}, or revealing structural interpretations that elucidate the fundamental mechanisms governing DD estimation performance. 
\end{enumerate}
As a result, the lack of a general and unified estimation-theoretic framework hampers systematic comparison and benchmarking of sensing performance across different waveforms and system configurations, and, more importantly, obscures structure-revealing insights into the CRB for joint DD estimation. These gaps motivate this work.

\emph{Major Contents and Contributions:} This paper's major contents and principal contributions are summarized as follows.
\begin{itemize}
	\item A unified analytical FIM framework is established under general sensing settings, which directly quantifies the joint DD CRBs across various waveforms. Particularly, different waveforms are modeled via orthonormal bases (i.e., unitary transformations) that map transmitted information-bearing data in different domains to time-domain (more specifically, represented by slow- and fast-time) symbols, thereby admitting a unified and common representation. The FIM analysis reveals that the sensing performance is fundamentally governed by the power profile of time-frequency (TF, i.e., in slow-time and Fourier transformed fast-time) domain symbols, which is shaped by the waveform-specific modulation.
	\item The key finding reveals that different waveforms, under unitary data-to-symbol modulation, can theoretically yield same DD CRBs under the same TF symbol-level power profile (corresponding to different data sequence). Conversely, waveform sensing advantage is data-dependent. If certain waveform yield ``better'' CRB, it is because that waveform-specific modulation reallocates the power of a \emph{particular choice} of the data sequence to a TF symbol-level power profile that is better aligned with the FIM structure.
	\item The Pareto boundary of the joint DD CRB region is identified by optimizing over the power profile of transmitted TF symbols. It reveals the system's joint DD estimation-theoretic performance limits, i.e., the system’s ultimate DD CRB capability, beyond some commonly adopted CRB expressions. An efficient and low-complexity numerical algorithm is developed to trace this boundary.
	\item Closed-form FIM/CRB expressions and optimal waveform structures are derived for the special single-target case. They enable analytical tractability and asymptotic analysis, providing interpretable insights into the fundamental CRB performance on joint DD estimation.
\end{itemize}

Additionally, all theoretical findings are validated with numerical examples. To the best of the authors' knowledge, this is the first framework to quantify the joint DD CRB performance for diverse multi-carrier waveforms, and the first work to explicitly identify the estimation-theoretic limits of the achievable multi-parameter (e.g., delay and Doppler) CRB region. The proposed framework is expected to serve as a versatile building block for further studies incorporating more sophisticated waveform-specific processing.


\emph{Notational Convention:} Regular, lowercase bold, and uppercase bold symbols respectively represent a scalar, vector, and matrix (e.g., $a$, $\mathbf{a}$, and $\mathbf{A}$). $\mathbb{C}$, $\mathbb{R}$ denote the complex and real sets; their dimension is indicated in superscripts. $\jj$ is the imaginary unit such that $\jj^2=-1$. $\I$ is the identity matrix, $\mathbf{1}$ (resp. $\mathbf{0}$) is the matrix with all $1$'s (resp. $0$'s); their dimensions are indicated in subscripts. $|\cdot|$, $\|\cdot\|$ denote the absolute value of a scalar, the Euclidian norm of a vector, respectively; $(\cdot)^{-1}$, $(\cdot)^*$, $(\cdot)^\trans$, $(\cdot)^\herm$, $\Re(\cdot)/\Im(\cdot)$, $\vect\{\cdot\}$, and $\Tr(\cdot)$ represent inverse, conjugate, transpose, Hermitian transpose, the real/imaginary part, the columnwise stacked vector, and the trace of a matrix, respectively. $\Mdiag\{\cdot\}$
returns a matrix formed by putting its arguments on its main diagonal entries or blocks, depending on the context; $\vdiag\{\cdot\}$ returns a vector formed using the diagonal elements of its (square) matrix argument. $\otimes$ is the Kronecker product. Let $\mathbf{A},\mathbf{B} \in \mathbb{C}^{PQ \times PQ}$ be block-partitioned matrices composed of $P \times P$ blocks of size $Q \times Q$, i.e., $\mathbf{A} = [\mathbf{A}_{\{p_1,p_2\}}]$; $\mathbf{B} = [\mathbf{B}_{\{p_1,p_2\}}]$, where $\mathbf{A}_{\{p_1,p_2\}},\mathbf{B}_{\{p_1,p_2\}} \in \mathbb{C}^{Q \times Q}$, $1\leq p_1,p_2\leq P$. $\Bodot$ is the block Hadamard product, defined as $$\mathbf{A}\Bodot\mathbf{B} := [\mathbf{A}_{\{p_1,p_2\}}\mathbf{B}_{\{p_1,p_2\}}],$$
i.e., the $(p_1,p_2)$-th $Q\times Q$ sub-block of the matrix $(\mathbf{A}\Bodot\mathbf{B})$ is $\mathbf{A}_{\{p_1,p_2\}}\mathbf{B}_{\{p_1,p_2\}}$. Finally, $\BTr(\cdot)$ is the block trace operator, defined as
\begin{equation*}
	\BTr(\mathbf{A})\! :=\! \begin{bmatrix}
		\Tr(\mathbf{A}_{\{1,1\}})& \Tr(\mathbf{A}_{\{1,2\}}) & \cdots & \Tr(\mathbf{A}_{\{1,P\}})\\
		\Tr(\mathbf{A}_{\{2,1\}}) & \Tr(\mathbf{A}_{\{2,2\}}) & \cdots & \Tr(\mathbf{A}_{\{2,P\}}) \\
		\vdots & \vdots & \ddots & \vdots  \\
		\Tr(\mathbf{A}_{\{P,1\}}) & \Tr(\mathbf{A}_{\{P,2\}}) & \cdots & \Tr(\mathbf{A}_{\{P,P\}})
	\end{bmatrix}\!.
\end{equation*}
\section{System Model}
\label{sec:system_model}

As shown in Fig. \ref{fig:system_model}, a generic sensing system, e.g., a radar\footnote{Although the exposition adopts a radar-oriented interpretation, the proposed framework is equally applicable to device-centric sensing, where the transmitted signal $\underline{s}(t)$ serves as a known probing sequence (e.g., pilot or reference signal) for estimating DD parameters in multipath channels.} setup, is considered. It operates at the frequency $\fc$ with a total bandwidth $B$. A high-dynamic time-domain sensing waveform $\underline{s}(t)$ is used to estimate the delays and Doppler shifts related to the targets. Here, $\underline{s}\in\{\underline{s}_{\mathsf{w}}|\mathsf{w}\in\mathcal{W}\}$ with $\mathcal{W}:=\{\OFDM,\OTFS,\ODDM,\AFDM,\OCDM\}$. 

Specifically, $\underline{s}(t)$ is sent by the transmitter (Tx), reflected by $P$ targets/scatters, and received by the receiver (Rx). The sensing operation can be conducted in either monostatic or bistatic mode. For the latter case, the Tx and Rx are synchronized. As a result, the Rx has the perfect knowledge of $\underline{s}(t)$ \cite{wang2025device,OFDM+OTFS,OFDM+OTFS2,braun2014ofdm,wang_2026_TIT} for both cases, as in standard device-free sensing settings.\footnote{As the focus is on sensing performance under different waveforms, and the signal $\underline{s}(t)$ is known for Rx signal processing, its randomness is not explicitly considered herein. When $\underline{s}(t)$ is random, the analysis remains applicable by adopting a Miller-Chang type CRB \cite{Miller1978MCRB} instead of the conventional CRB, where the estimation accuracy is expected over the signal distribution.}

\begin{figure}[t]
	\centering
	\includegraphics[scale=1]{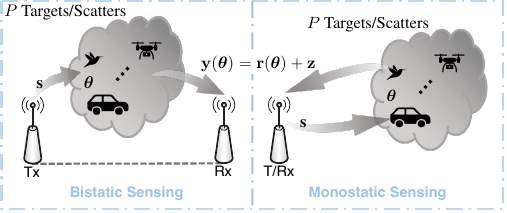}
	\caption{A notional sensing system considered in this paper.}
	\label{fig:system_model}
\end{figure}

Each scatter is regarded as a point target \cite{wang2025device,braun2014ofdm,OFDM+OTFS}, so that the sensing link between T/Rx and targets can be modeled as $P$-tap (multipath) TF selective channels:
\begin{equation}
	\gtf(t,\tau;\boldtheta) = \sum_{p=1}^{P} g_p\delta(\tau-\tau_p){\rm e}^{\jj 2\pi\nu_p t},
	\label{eqn:channel_tf}
\end{equation}
where $g_p\in\mathbb{C}$, $\nu_p\in\mathbb{R}$, and $\tau_p\in\mathbb{R}$ denote the channel gain (e.g., combined pathloss and target reflection), the Doppler shift, and the delay of the $p$-th tap, respectively; $\boldtheta$ is the collection of all unknown parameters
\begin{equation*}
	\boldtheta:=\Big[\boldtau^\trans,\boldnu^\trans,\g^\trans\Big]^\trans \in \mathbb{R}^{4P},
\end{equation*}
where $\boldtau:=[\tau_1,\cdots,\tau_P]^\trans\in\mathbb{R}^{P}$, $\boldnu:=[\nu_1,\cdots,\nu_P]^\trans\in\mathbb{R}^{P}$, and $\g:=[\gr^\trans,\gi^\trans]^\trans\in\mathbb{R}^{2P}$, with $\gr:=\Re\big\{[g_1,\cdots,g_P]^\trans\big\}$ and $\gi:=\Im\big\{[g_1,\cdots,g_P]^\trans\big\}$. Here, $\boldtau$ and $\boldnu$ are of interest; $\g$ is of nuisance. In the DD domain, the channel \eqref{eqn:channel_tf} is
\begin{equation*}
	\gdd(\nu,\tau;\boldtheta) = \mathscr{F}_{t}\Big\{\gtf(t,\tau;\boldtheta)\Big\} = \sum_{p=1}^P g_p\delta(\tau-\tau_p)\delta(\nu-\nu_p).
\end{equation*}
Here, $\mathscr{F}_t\{\cdot\}$ is the Fourier transform with respect to  $t$. 

Consider a block-transmission system where the transmitted time-domain symbol is represented by $\S\in\mathbb{C}^{M\times N}$. Here, $M$ and $N$ denote the number of samples in the fast- and slow-time domain, which determine delay and Doppler resolutions, respectively. Under such transmission, the sample duration is $\Ts:=1/B$.\footnote{In this paper, Nyquist sampling is primarily considered \cite{FanLiu_TIT,OFDM+OTFS2} to gain general insights. Faster-than-Nyquist sampling can be directly built upon this framework by setting $T_{\mathsf{s},\mathsf{Nyquist}}=1/B$ and $\Ts=T_{\mathsf{s},\mathsf{Nyquist}}/L$, where $L>1$ is the oversampling ratio.} By adding a cyclic prefix (CP) or chirp-periodic prefix (CPP)\footnote{For waveforms like AFDM, a CPP \cite{AFDM_communication} is used instead of a CP due to different signal periodicity. For either prefixing, they induce a circulant matrix structure in the effective delay channel. As a result, \eqref{eqn:delay_circular} holds in general.} with the prefix length $\Mcp$ that is no less than the maximum delay spread in the channel \cite{FanLiu_iceberg_TSP}, the time-domain symbol $\S$ becomes $\Scp\in\mathbb{C}^{(M+\Mcp)\times N}$. The transmitted continuous-time signal is thus
\begin{equation*}
	\underline{s}(t) = \sum_{\tilde{n}=0}^{N-1}\sum_{\tilde{m}=-\Mcp}^{M-1} [\Scp]_{\tilde{m}',\tilde{n}'} \psi_{\mathsf{ps}}(t-\tilde{m}\Ts-\tilde{n}\Tsym),
\end{equation*}
where $\tilde{m}'=\tilde{m}+1$, $\tilde{n}'=\tilde{n}+1$, $\Tsym := (M+\Mcp)\Ts$ is the prefix-extended duration for one slow-time sample, and $\psi_{\mathsf{ps}}(\cdot)$ represents a pulse shape with unit power. Here, $\mathsf{ps}\in\{\mathsf{rect},\mathsf{rrc},\mathsf{sinc},\mathsf{Gaussian}\}$, which respectively represent rectangular, root raised cosine (RRC), sinc, and Gaussian pulses. After passing through the the channel \eqref{eqn:channel_tf}, the received noise-free signal is
\begin{equation*}
	\underline{r}(t;\boldtheta) = \int \gtf(t,\tau;\boldtheta)s(t-\tau)\ \mathrm{d}{\tau} =  \sum_{p=1}^P g_p\underline{s}(t-\tau_p){\rm e}^{\jj 2\pi\nu_p t}.
\end{equation*}
\section{Cram\'{e}r-Rao Bound on Joint Delay-Doppler Estimation}
\label{sec:CRB}


\subsection{Measurement Model}
Denote $\Rcp\in\mathbb{C}^{(M+\Mcp)\times N}$ as the received noise-free sample after the matched filter for $\psi_{\mathsf{ps}}(t)$ and sampled at discrete time $t=m\Ts+n\Tsym$. It is explicitly given as
\begin{equation}
	\begin{aligned}
		&[\Rcp]_{m',n'} = \Big[\int \underline{r}(\xi;\boldtheta)\psi_{\mathsf{ps}}^*(\xi-t) \ \mathrm{d} \xi\Big]\Big|_{t=m\Ts+n\Tsym} \\
		=&\ \sum_{p=1}^{P}g_p\sum_{\tilde{n}=0}^{N-1}\sum_{\tilde{m}=-\Mcp}^{M-1} [\Scp]_{\tilde{m}',\tilde{n}'} {\rm e}^{\jj 2\pi\nu_p n\Tsym}  \\
		&\hspace{1.6cm}\times\Delta_{\mathsf{ps}}[(m-\tilde{m})\Ts+(n-\tilde{n})\Tsym-\tau_p],
	\end{aligned}
	\label{eqn:measurement_model_discrete_withCP}
\end{equation}
where $m'=m+1$, $n'=n+1$, and $\Delta_{\mathsf{ps}}(t):=\int \psi_{\mathsf{ps}}(\xi)\psi_{\mathsf{ps}}^*(\xi-t)\ \mathrm{d}{\xi}$ is the auto-correlation of transmit pulse.
\footnote{Here, the variation of Doppler shift within one symbol duration $\Tsym$ is neglected, so that it only depends on the slow-time index $n$ \cite{OFDM+OTFS,OFDMRadar_IEEE_Proc,braun2014ofdm}. In this way, the Doppler shift and the delay are decoupled, simplifying the estimation process and facilitating general analytical exposition. Moreover, the model has been proven to accurately capture the Doppler effect when $B\ll\fc$ \cite{OFDMRadar_IEEE_Proc}, or when the $\max\big\{\{\nu_{p}\}_{p=1}^P\big\}\ll1/\Tsym$ \cite{OFDM+OTFS}. These conditions are easily met in practical mobile networks \cite{OFDM+OTFS}.}
After sampling, the Rx discards the first $\Mcp$ samples in each slow-time symbol:
\begin{equation*}
	[\R]_{m',n'} = [\Rcp]_{m',n'},\quad 1\leq m'\leq M,\ 1\leq n'\leq N.
\end{equation*}
Due to the presence of the CP/CPP, the inter-symbol interference can be prevented and the contribution comes from $\tilde{n}=n$ in \eqref{eqn:measurement_model_discrete_withCP} \cite{FanLiu_iceberg_TSP}. The received noise-free samples are
\begin{equation*}
	[\R]_{m',n'} \!=\! \sum_{p=1}^{P}g_p\!\!\sum_{\tilde{m}=0}^{M-1} [\S]_{\tilde{m}',n'} {\rm e}^{\jj 2\pi\nu_p n\Tsym} \Delta_{\mathsf{ps}}[(m-\tilde{m})\Ts-\tau_p].
\end{equation*}
Coherently, the measurement model can be written as
\begin{equation}
	\Y(\boldtheta) = \R + \Z = \sum_{p=1}^Pg_p\boldDelta_{\mathsf{ps}}(\tau_p)\S\V(\nu_p)+\Z,
	\label{eqn:measurement_model_intermediate}
\end{equation}
where $\boldDelta_{\mathsf{ps}}(\cdot)$ is a Toeplitz matrix, given as $[\boldDelta_{\mathsf{ps}}(\tau)]_{m_1',m_2'}=\Delta_{\mathsf{ps}}((m_1'-m_2')\Ts-\tau)$, $1\leq m'_1,m'_2\leq M$, 
\begin{equation*}
	\V(\nu_p) := {\Mdiag}\Big\{\Big[{\rm e}^{\jj 2\pi\nu_pn\Tsym}\Big]_{n=0}^{N-1}\Big\},
\end{equation*}
and $\Z$ is the sampled additive noise. When the prefix length $\Mcp$ is no smaller than the maximum delay spread in the channel, $\boldDelta_{\mathsf{ps}}(\tau_p)$ becomes a circulant matrix, which is diagonalizable by the discrete Fourier transform (DFT) matrix, i.e.,
\begin{equation}
	\boldDelta_{\mathsf{ps}}(\tau_p) =\F_M^\herm \T_{\mathsf{ps}}(\tau_p)\F_M,
	\label{eqn:delay_circular}
\end{equation}
where $\F_M$ is the DFT matrix of degree $M$, given as
\begin{equation*}
	[\F_M]_{m_1',m'_2} = \frac{1}{\sqrt{M}}{\rm e}^{-\jj \frac{2\pi}{M}(m'_1-1)(m'_2-1)}, \ \ 1\leq m'_1,m'_2\leq M,
\end{equation*}
and $\T_{\mathsf{ps}}(\tau_p)$ is a diagonal matrix, with its diagonal being the DFT of the first column of $\boldDelta_{\mathsf{ps}}(\tau_p)$. More specifically, 
\begin{equation*}
	\begin{aligned}
	& [\T_{\mathsf{ps}}(\tau_p)]_{m',m'} = 				\sum_{\tilde{m}=0}^{M-1}\Delta_{\mathsf{ps}}(\tilde{m}\Ts-\tau_p)\mathrm{e}^{-\jj\frac{2\pi m\tilde{m}}{M}} \\
		=&\ \frac1{\Ts}\sum_{i=-\infty}^\infty  \Big|\Psi_{\mathsf{ps}}\Big(\frac1{\Ts}\Big(\frac{m}{M}+i\Big)\Big)\Big|^2\mathrm{e}^{-j2\pi\frac{m+iM}{M}\frac{\tau_p}{\Ts}},
	\end{aligned}
\end{equation*}
where the second equality comes from the Poisson summation, and $\Psi_{\mathsf{ps}}(f):=\mathscr{F}_t\{\psi_{\mathsf{ps}}(t)\}$, which is explicitly evaluated as
\begin{equation*}
	\begin{aligned}
		&\Psi_{\mathsf{rect}}(f) = \sqrt{\Ts}\sinc(f\Ts)\mathrm{e}^{-j\pi f\Ts} ; \\
		&\Psi_{\mathsf{rrc}}(f) = \sqrt{\Ts}\mathrm{e}^{-j\pi f\Ts}\\
		& \quad\quad\quad\quad \times \begin{cases}
			1, & |f|\leq\frac{1-\beta}{2\Ts}\\
			\sqrt{\frac{1+\cos(\frac{\pi}{\beta}(\Ts|f|-\frac{1-\beta}{2}))}{2}}, & \frac{1-\beta}{2\Ts}<|f|\leq\frac{1+\beta}{2\Ts}\\
			0, & |f|>\frac{1+\beta}{2\Ts}
		\end{cases};\\
		&\Psi_{\mathsf{sinc}}(f) = \sqrt{\Ts} \mathrm{rect}(f\Ts)\mathrm{e}^{-j\pi f\Ts};\\
		&\Psi_{\mathsf{Gaussian}}(f) = \sqrt{\Ts}\Big(\frac{2\pi}{\gamma}\Big)^{1/4}{\exp}\Big(-\frac{\pi^2 f^2\Ts^2}{\gamma}\Big)\mathrm{e}^{-j\pi f\Ts},\\
	\end{aligned}
\end{equation*}
where $\beta\geq 0$ is the rolling parameter for the RRC pulse, and $\gamma$ is the expansion parameter for the Gaussian pulse.

\begin{figure}[t]
	\centering
	\includegraphics*[scale=1]{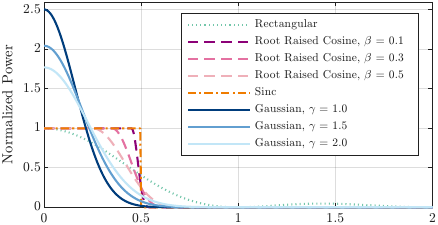}\\
	\vspace{-0.15cm}
	{\footnotesize $\quad \quad |f\Ts|$}
	\caption{$|\Psi_{\mathsf{ps}}(f)|^2/\Ts$ as a function of $|f\Ts|$ for various $\mathsf{ps}$.}
	\label{fig:Pulse_shaping}
\end{figure}

For practical pulse-shaping filters, the energy is predominantly concentrated within the main lobe satisfying $|f\Ts|\leq 1/2$ (see Fig. \ref{fig:Pulse_shaping}). Given the discrete normalized frequency index $m/M$ falls strictly within $[0, 1)$, to capture the significant energy of the pulse --- which requires $|m/M+i|\leq 1/2$ in the Poisson summation, the summation index is effectively restricted to $i=0$ for $0\leq m/M<1/2$ and to $i=-1$ for $1/2\leq m/M<1$.
For other values of $i$, the corresponding sampling points fall into the regions where the energy experiences rapid polynomial or exponential decay, yielding negligible contributions to the overall summation. Therefore, to ensure the analytical tractability of the subsequent derivations, it is well-justified to neglect these high-order spectral aliasing terms and retain only the dominant energy components (alternatively, one can keep all $i\neq 0,-1$ terms and compute numerically). Consequently,
\begin{equation*}
	{\T_{\mathsf{ps}}}(\tau_p) := {\Mdiag}\Big\{\Big[\Upsilon_{\mathsf{ps}}(m){\rm e}^{-\jj{2\pi}\eta_m\frac{\tau_p}{\Ts} }\Big]_{m=0}^{M-1}\Big\},
\end{equation*}
where $\Upsilon_{\mathsf{ps}}(m):=\frac1{\Ts}\big|\Psi_{\mathsf{ps}}(\frac1{\Ts}\eta_m)\big|^2$, and
$$\eta_m:=\begin{cases}
	m/M,& 0\leq m< \lfloor M/2\rfloor\\
	m/M-1, & \lfloor M/2\rfloor\leq m\leq M-1
\end{cases}.$$

Finally, the measurement model \eqref{eqn:measurement_model_intermediate} becomes
\begin{equation*}
	\begin{aligned}
		\Y(\boldtheta)&= \F_M^\herm\bigg[\sum_{p=1}^P g_p\T_{\mathsf{ps}}(\tau_p)\F_M\S\V(\nu_p)\bigg] + \Z.
	\end{aligned}
\end{equation*}
To provide an equivalent and more tractable characterization of the FIM in subsequent exposition, the above equation is left-multiplied with $\F_M$:
$$\Big(\F_M\Y\Big)= \sum_{p=1}^P g_p\T_{\mathsf{ps}}(\tau_p)\Big(\F_M\S\Big)\V(\nu_p) + \Big(\F_M\Z\Big).$$
As $\Y$, $\S$, and $\Z$ represent $2$-dimensional symbols in slow- and fast-time domains, $(\F_M\Y)$, $(\F_M\S)$, and $(\F_M\Z)$ represent $2$-dimensional TF domain symbols. Specifically, these TF domain observation is obtained by i) applying a DFT along the fast-time columns and then ii) stacking the transformed fast-time columns. Denote $\y:=\vect\{\F_M\Y\}$, $\z:=\vect\{\F_M\Z\}$, and $\s:=\vect\{\F_M\S\}$ the TF-domain symbol vectors, the measurement model is equivalent to
\begin{equation}
	\label{eqn:measurement_model}
	\begin{aligned}
		&\y(\boldtheta) = \bigg[\sum_{p=1}^P g_p\V(\nu_p)\otimes\T_{\mathsf{ps}}(\tau_p)\bigg]\s+\z \\
		=&\ \NU(\boldnu)\G(\g)\TAU(\boldtau)\s + \z =\r(\boldtheta)+\z,
	\end{aligned}
\end{equation}
where $\r(\boldtheta):=\NU(\boldnu)\G(\g)\TAU(\boldtau)\s$, with
\begin{equation*}
	\begin{alignedat}{2}
		\NU(\boldnu) &:= [\V(\nu_1)\otimes\I_M,\cdots,\V(\nu_P)\otimes\I_M]  &\in \mathbb{C}^{Q\times PQ} ; \\
		\TAU(\boldtau) &:= [\I_N\otimes\T_{\mathsf{ps}}(\tau_1),\cdots,\I_N\otimes\T_{\mathsf{ps}}(\tau_P)]^\trans
		&\in\mathbb{C}^{PQ\times Q}; \\
		\G(\g) &:= {\Mdiag}\big\{[g_1,g_2,\cdots,g_P]^\trans\big\}\otimes\I_Q   &\in\mathbb{C}^{PQ\times PQ}.
	\end{alignedat}
\end{equation*}
Here, $Q:=MN$ is the total number of samples in a processing interval. Denote $\mathcal{Q}:=\{1,2,\cdots,Q\}.$ Following conventions of multi-carrier waveforms for sensing/ISAC,  additive white Gaussian noise (AWGN) is considered herein \cite{FanLiu_TIT,FanLiu_iceberg_TSP,OFDM+OTFS,OFDM+OTFS2}, i.e., $\z\sim\mathcal{CN}(\bm{0}_Q,\noisevar\mathbf{I}_Q)$, where $\noisevar$ is the noise variance.

\begin{remark}
	As a consequence of the receiver front-end consisting of a pulse-matched filter followed by sampling at the chip rate $\Ts$, the discrete-time noise $\z$ is AWGN. Specifically, the output noise of the matched filter is a zero-mean Gaussian process whose auto-correlation function is determined by the auto-correlation of the receive pulse. When the matched filter satisfies the Nyquist criterion, the auto-correlation evaluated at integer multiples of $\Ts$ is zero for all nonzero lags. Consequently, the noise samples obtained by sampling at the chip rate are mutually uncorrelated, resulting in an AWGN in discrete time.\footnote{For a Gaussian pulse, the Nyquist criterion is not satisfied exactly; nevertheless, it can be approximately regarded as quasi-Nyquist when the residual intersample correlation is negligible.}
\end{remark}

\begin{remark}
	Beyond \eqref{eqn:measurement_model}, different signal processing may be further involved at the Rx side for different waveforms. As subsequent signal processing would \emph{not} create any additional (Fisher) information when $\S$ (hence, $\s$) is known at the Rx \cite{range_compression_TSP,Kay1993Fundamentals}, further waveform-specific processing are not explicitly discussed and \eqref{eqn:measurement_model} is employed for generic analysis.
\end{remark}

\begin{definition}
	Under AWGN, the FIM for estimating $\boldtheta$ is (the Slepian-Bangs formula, \cite[Theorem 2.1]{Complex_Stochastic_Processes})
	\begin{equation*}
		\begin{aligned}
			\FIM_{\boldtheta\boldtheta} :=&\  2\Re\Bigg\{ \Big(\pd{\r(\boldtheta)}{\boldtheta^\trans}\Big)^\herm (\noisevar\mathbf{I}_Q)^{-1} \Big(\pd{\r(\boldtheta)}{\boldtheta^\trans}\Big) \Bigg\} \\
			=& \frac{2}{\noisevar} \Re\left\{ \pd{\r^\herm(\boldtheta)}{\boldtheta}\pd{\r(\boldtheta)}{\boldtheta^\trans} \right\} \in\mathbb{R}^{4P\times 4P}.
		\end{aligned}
	\end{equation*}
	In the sequel, the structure of the FIM will be explored.
\end{definition}

\subsection{Structure and Evaluation of the FIM and CRB}

\begin{theorem} \label{theorem:FIM_structure}
	The FIM $\FIM_{\boldtheta\boldtheta}$ can be block-partitioned as
	\begin{equation*}
	\FIM_{\boldtheta\boldtheta} \!=\!	\renewcommand{\arraystretch}{1.15}
		\frac{2}{\noisevar}\!\!
		\left[\!\!\!
		\begin{array}{rrrr}
			\Re\{\J_{\boldtau\boldtau}\}& \Re\{\J_{\boldtau\boldnu}\} & \Re\{\J_{\boldtau\g}\} & -\Im\{\J_{\boldtau\g}\}\\
			\Re\{\J^\trans_{\boldtau\boldnu}\} & \Re\{\J_{\boldnu\boldnu}\} & \Re\{\J_{\boldnu\g}\} & -\Im\{\J_{\boldnu\g}\}\\
			\Re\{\J^\trans_{\boldtau\g}\} & \Re\{\J^\trans_{\boldnu\g}\} & \Re\{\J_{\g\g}\} & -\Im\{\J_{\g\g}\} \\
			-\Im\{\J^\trans_{\boldtau\g}\} & -\Im\{\J^\trans_{\boldnu\g}\} & -\Im\{\J^\trans_{\g\g}\} &
			\Re\{\J_{\g\g}\}
		\end{array}\!\!\!
		\right]\!\!,
		\renewcommand{\arraystretch}{1}
	\end{equation*}
	where
	\begin{subequations}
		\label{eqn:J}
		\begin{align}
			\J_{\boldtau\boldtau} &:= Q\BTr\Big[\NU^\herm\NU\Bodot\G^*\dot{\bm{\mathcal{T}}}_{\!\!\mathsf{ps}}^*\Rs^\trans\dot{\bm{\mathcal{T}}}_{\!\!\mathsf{ps}}^\trans\G\Big]; \label{eqn:J_tau} \\
			\J_{\boldtau\boldnu} &:= Q\BTr\Big[ \NU^\herm\dot{\NU} \Bodot \G^*\dot{\bm{\mathcal{T}}}_{\!\!\mathsf{ps}}^*\Rs^\trans \TAU^\trans\G \Big]; \label{eqn:J_taunu}\\
			\J_{\boldnu\boldnu} &:=  Q\BTr\Big[\dot{\NU}^\herm\dot{\NU}\Bodot\G^*\TAU^*\Rs^\trans\TAU^\trans\G\Big]; \label{eqn:J_nu} \\
			\J_{\boldtau\g} &:= Q\BTr\Big[\NU^\herm\NU\Bodot\G^*\dot{\bm{\mathcal{T}}}_{\!\!\mathsf{ps}}^*\Rs^\trans\TAU^\trans\Big]; \label{eqn:J_taug} \\
			\J_{\boldnu\g} &:= Q\BTr\Big[\dot{\NU}^\herm\NU\Bodot\G^*\TAU^*\Rs^\trans\TAU^\trans\Big]; \label{eqn:J_nug} \\
			\J_{\g\g} &:=  Q\BTr\Big[\NU^\herm\NU\Bodot\TAU^*\Rs^\trans\TAU^\trans\Big], \label{eqn:J_g}
		\end{align}
	\end{subequations}
	with
	\begin{equation*}
		\begin{aligned}
			\dot{\NU} &:= \left[\dot{\V}(\nu_1)\otimes \I_M,\cdots,\dot{\V}(\nu_P)\otimes \I_M\right]; \\
			\dot{\bm{\mathcal{T}}}_{\!\!\mathsf{ps}} &:= \left[\I_N\otimes \dot{\T}_{\mathsf{ps}}(\tau_1),\cdots,\I_N\otimes \dot{\T}_{\mathsf{ps}}(\tau_P)\right]^\trans;\\
			\dot{\V}(\nu_p) :=&\ \pd{\V(\nu_p)}{\nu_p} =  {\Mdiag}\Big\{\Big[ \jj2\pi n\Tsym {\rm e}^{\jj 2\pi\nu_pn\Tsym} \Big]_{n=0}^{N-1}\Big\}; \\
			\dot{\T}_{\mathsf{ps}}(\tau_p) :=&\ \pd{\T_{\mathsf{ps}}(\tau_p)}{\tau_p}\\
			=&\ {\Mdiag}\Big\{\Big[-\jj\Upsilon_{\mathsf{ps}}(m)\frac{2\pi \eta_m}{\Ts}{\rm e}^{-\jj{2\pi}\eta_m\frac{\tau_p}{\Ts}}\Big]_{m=0}^{M-1}\Big\},
		\end{aligned}
	\end{equation*}	
	and  the transmit \emph{sample} correlation matrix is defined as 
	\begin{equation*}
		\Rs := \frac1Q \s\s^\herm.
	\end{equation*}
	
	\begin{IEEEproof}
		See Appendix \ref{Appendix:FIM_structure_proof} in Supplementary Information.
	\end{IEEEproof}
\end{theorem}

Theorem \ref{theorem:FIM_structure} presents a principled way to calculate the complete FIM corresponding to all unknown parameters. For numerical computation, \eqref{eqn:J} can be further expressed in an element-wise manner in Appendix \ref{appendix:FIM_element_wise} in Supplementary Information. From the calculated FIM, the CRB can be defined.

\begin{definition} \label{definition:CRB}
The CRB matrix for estimating $\boldtheta$ is
\begin{equation*}
	\CRB_\boldtheta := \FIM_{\boldtheta\boldtheta}^{-1}.
\end{equation*}
The CRB matrices for delay and Doppler are, respectively
\begin{equation*}
	\CRB_{\boldtau} := [\CRB_{\boldtheta}]_{1:P,1:P};\quad \CRB_{\boldnu} := [\CRB_{\boldtheta}]_{P+1:2P,P+1:2P}.
\end{equation*}
For any weakly-biased or unbiased estimator of $\boldtau$, denoted as $\hat{\boldtau}$ (resp. $\hat{\boldnu}$ for $\boldnu$), the following inequalities on the estimation mean square error (MSE) hold
\begin{equation*}
	\mathbb{E}\big\{(\boldtau-\hat{\boldtau})(\boldtau-\hat{\boldtau})^\trans\big\} \succeq \CRB_{\boldtau};\ \
	\mathbb{E}\big\{(\boldnu-\hat{\boldnu})(\boldnu-\hat{\boldnu})^\trans\big\} \succeq \CRB_{\boldnu},
\end{equation*}
or equivalently,
\begin{equation} \label{eqn:CRB_tau/nu}
\begin{aligned}
	\mathbb{E}\big\{\|\boldtau-\hat{\boldtau}\|^2\big\}\geq \Tr(\CRB_{\boldtau}):=\varepsilon(\boldtau);\\
	\mathbb{E}\big\{\|\boldnu-\hat{\boldnu}\|^2\big\}\geq \Tr(\CRB_{\boldnu}):=\varepsilon(\boldnu).\\
\end{aligned}
\end{equation}
The CRBs for each delay and Doppler element are respectively $\varepsilon(\tau_p) := [\CRB_{\boldtau}]_{p,p}$ and $\varepsilon(\nu_p) := [\CRB_{\boldnu}]_{p,p}$.
\end{definition}

\begin{proposition}
	\label{proposition:dependence_diagonal}
	From element-wise expressions in Appendix \ref{appendix:FIM_element_wise} and the property of the $\Tr(\cdot)$ operator, it is obtained that
	\begin{equation*}
		\begin{aligned} 		
			[\J&_{\boldnu\boldnu}]_{p_1,p_2} = Qg_{p_1}^*g_{p_2}\times \\
			&\sum_{q'\in\mathcal{Q}}\big[\dot{\V}^\herm(\nu_{p_1})\dot{\V}(\nu_{p_2})\otimes\T_{\mathsf{ps}}(\tau_{p_2})\T_{\mathsf{ps}}^*(\tau_{p_1})\big]_{q',q'} [\Rs^\trans]_{q',q'},
		\end{aligned}
	\end{equation*}
	where $q':=q+1$ for simplicity. The above equality holds as $\dot{\V}^\herm(\nu_{p_1})\dot{\V}(\nu_{p_2})\otimes\T_{\mathsf{ps}}(\tau_{p_2})\T_{\mathsf{ps}}^*(\tau_{p_1})$ is diagonal. It can be similarly shown that each element in $\J_{\boldtau\boldtau}$, $\J_{\boldtau\boldnu}$, $\J_{\boldtau\g}$, $\J_{\boldnu\g}$, $\J_{\bm{gg}}$, and hence, $\FIM_{\boldtheta\boldtheta}$, depends only on the diagonal terms of $\Rs$, i.e., the \textit{power profile} of the transmitted TF symbol, suggesting that cross sample correlations (i.e., cross products) among TF symbols do not influence the CRB on joint DD estimation. 
\end{proposition}

To provide further engineering insights, the following theorem presents closed-form evaluations of the Fisher information and CRB matrix for the single-target scenario.

\begin{theorem}[FIM for single-target DD estimation]
	\label{theorem:single_target_general}
	For single-target scenario, i.e., $P=1$, $\tau=\tau_1$, $\nu=\nu_1$, $g=g_1$, \eqref{eqn:J} admits following evaluations ($[\s]_{q'}$ is the $q'$-th element in vector $\s$)
	\begin{subequations}
		\label{eqn:J_scalar}
		\begin{align}
			\mathscr{J}_{\tau\tau} &= |g|^2\sum_{q'\in\mathcal{Q}}\Upsilon_{\mathsf{ps}}^2(m)\Big(\frac{2\pi\eta_ m}{\Ts}\Big)^2|[\s]_{q'}|^2; \label{eqn:J_tau_scalar} \\
			\mathscr{J}_{\tau\nu} &= - |g|^2\sum_{q'\in\mathcal{Q}}\Upsilon_{\mathsf{ps}}^2(m)(2\pi  n\Tsym)\Big(\frac{2\pi \eta_m}{\Ts}\Big)|[\s]_{q'}|^2;
			\label{eqn:J_taunu_scalar} \\
			\mathscr{J}_{\nu\nu} &= |g|^2\sum_{q'\in\mathcal{Q}} \Upsilon_{\mathsf{ps}}^2(m)(2\pi n\Tsym)^2|[\s]_{q'}|^2; \\
			\mathscr{J}_{\tau g} &= \jj g^*\sum_{q'\in\mathcal{Q}}\Upsilon_{\mathsf{ps}}^2(m) \Big(\frac{2\pi \eta_m}{\Ts}\Big)|[\s]_{q'}|^2; \\
			\mathscr{J}_{\nu g} &= -\jj g^*\sum_{q'\in\mathcal{Q}} \Upsilon_{\mathsf{ps}}^2(m) (2\pi n\Tsym)|[\s]_{q'}|^2;\\
			\mathscr{J}_{gg} &= \sum_{q'\in\mathcal{Q}} \Upsilon_{\mathsf{ps}}^2(m) |[\s]_{q'}|^2.
			\label{eqn:J_g_scalar}
		\end{align}
	\end{subequations}
	The CRB matrix for joint estimation of $\tau$ and $\nu$ satisfies
	\begin{equation}
		\label{eqn:EFIM_single_target}
		\CRB_{\{\tau,\nu\}} = \renewcommand{\arraystretch}{1.6}
		\frac\noisevar2\left[\begin{array}{cc}
			\mathscr{J}_{\tau\tau}-\frac{|\mathscr{J}_{\tau g}|^2}{|\mathscr{J}_{gg}|^2} &
			\mathscr{J}_{\tau\nu}+\frac{|\mathscr{J}_{\tau g}\mathscr{J}_{\nu g}|}{|\mathscr{J}_{gg}|^2} \\
			\mathscr{J}_{\tau\nu}+\frac{|\mathscr{J}_{\tau g}\mathscr{J}_{\nu g}|}{|\mathscr{J}_{gg}|^2} &
			\mathscr{J}_{\nu\nu}-\frac{|\mathscr{J}_{\nu g}|^2}{|\mathscr{J}_{gg}|^2}
		\end{array}\right]^{-1}.
	\end{equation}
	
	\begin{IEEEproof}
		See Appendix \ref{appendix:single_target_general} in Supplementary Information. 
	\end{IEEEproof}
\end{theorem}

\begin{remark}
	It can be observed that the FIM is contributed from all $Q$ TF samples in an additive manner. It should be further noted that, under the single-target scenario, the FIM, and hence the CRB, do not depend on the specific values of the parameters of interest, e.g., $\tau$ and $\nu$. This observation is fundamentally different from the general multi-target case, where the FIM and CRB typically depend on $\boldtau$ and $\boldnu$. 
\end{remark}

\subsection{Impacts of Waveform Structure}
In this part, this paper investigates how waveform structures influence delay-Doppler estimation performance. Denote $\X\in\mathbb{C}^{M\times N}$ as the sensing information (e.g., data or information) placed in certain domain, the mappings between time-domain symbol $\S$ and $\X$ for various waveforms are: 
\begin{itemize}[leftmargin=*]
	\item {\bf OFDM}: $\X$ is placed in the TF domain (each column of $\X$ is in frequency-domain), and $\S$ is a time-domain OFDM symbol with $M$ subcarriers and $N$ slow-time samples, e.g., $\S_\OFDM=\F_M^\herm\X$.
	\item {\bf OTFS}: $\X$ is placed in the DD domain, it is first transformed into the TF domain through the inverse symplectic finite Fourier transform, e.g., $\F_M\X\F_N^\herm$; this TF symbol will then undergo the Heisenberg transform (row-wise inverse DFT) to the time domain \cite{ISAC_OTFS_TWC}: $\S_\OTFS=\F_M^\herm\F_M\X\F_N^\herm=\X\F_N^\herm$.
	\item {\bf ODDM}: $\X$ is placed in the DD domain, it is transformed into the time domain by $N$-point inverse DFT, i.e., $\S_\ODDM=\X\F_N^\herm$ \cite{ODDM_ISAC}.\footnote{Under the general framework in this paper and \cite{FanLiu_TIT}, the discrete-time mathematical formulations of OTFS and ODDM become identical. Such a formulation allows to treat them as equivalent within the scope of considered CRB analysis, without delving into implementation-specific details, which could differentiate them practically. 
	}
	\item {\bf AFDM}: $\X$ is placed in the affine Fourier transform domain, and is transformed into the time domain through the inverse discrete affine Fourier transform: $\S_\AFDM=\boldLambda_{c_1}^\herm\F_M^\herm\boldLambda_{c_2}^\herm\X$ \cite{AFDM_communication,FanLiu_TIT,AFDM_ISAC_MIMO}, where
	 	\begin{equation*}
	 		\boldLambda_{c}:={\Mdiag}\Big\{\Big[{\rm e}^{-\jj 2\pi cm^2}\Big]_{m=0}^{M-1}\Big\}.
	 	\end{equation*}
 	\item {\bf OCDM}: $\X$ is placed in the orthogonal chirp transform domain, and is transformed into the time domain through the discrete Fresnel transform: $\S_\OCDM= \boldXi^\herm\X$ \cite{OCDM_original}, where 
		\begin{equation*}
			[\boldXi]_{m_1',m_2'}=\frac1{\sqrt{M}}{\rm e}^{-\jj\frac\pi4}\begin{cases}
				{\rm e}^{\jj\frac{\pi}{M}(m_1'-m_2')^2} & \text{if }M\text{ is even} \\
				{\rm e}^{\jj\frac{\pi}{M}(m_1'+\frac12-m_2')^2} & \text{if }M\text{ is odd}
			\end{cases}.
		\end{equation*}
\end{itemize}
Denote $\x=\vect\{\X\}\in\mathbb{C}^Q$ as the data vector. According to $\s = \vect\{\F_M\S\}$ and the property of $\vect\{\cdot\}$ operator, the general matrix-based relationship between TF symbol $\s$ and data  $\x$ can be written as $\sw= \Uw\x$ for $\mathsf{w}\in\mathcal{W}$, with
\begin{equation*}
	\begin{aligned}
		\UOFDM &= \I_Q; \\
		\UOTFS &= \F_N^\herm\otimes \F_M; \\
		\UODDM &= \F_N^\herm\otimes \F_M;\\
		\UAFDM &= \I_N\otimes(\F_M\boldLambda_{c_1}^\herm\F_M^\herm\boldLambda_{c_2}^\herm); \\
		\UOCDM &= \I_N\otimes(\F_M\boldXi^\herm).
	\end{aligned}
\end{equation*}
Define the \emph{sample} data correlation matrix $\Rx:=\frac1Q\x\x^\herm$, it is hence related to $\Rsw$ with
\begin{equation}
	\Rsw = \Uw\Rx\Uw^\herm.
	\label{eqn:Rx-Rs}
\end{equation}
The FIM can be directly computed by plugging \eqref{eqn:Rx-Rs} into \eqref{eqn:J} for each $\mathsf{w}$ to assess the waveform-specific DD CRB for given $\Rx$. In subsequent analysis, this paper considers the transmitted waveform under a per-symbol power constraint $\pt$:
\begin{equation}
	\Tr(\Rx)=\pt={\|\x\|^2}/{Q}.
	\label{eqn:power_constraint}
\end{equation}
As $\Uw$ is unitary, e.g., $\Uw^\herm\Uw=\I_Q$, all waveforms yield the same power, i.e., $\Tr(\Rsw)=\pt={\|\s\|^2}/{Q}$.

\begin{remark}
	As the structure of FIMs (and CRBs) depends on the underlying measurement model and the power profile of transmitted TF symbols,
	following key facts are observed:
	\begin{itemize}
		\item[i)] For different waveforms (possibly with different data sequence $\xw$), as long as they share the same power profiles among TF symbols, i.e., the same $\vdiag\{\Rsw\}$, 
		their CRB performance on joint DD estimation will be identical. For example, 
		\begin{itemize}
			\item The ``optimal'' CRBs are identical for different waveforms. For any optimal $\Rs^\star$ corresponding to certain CRB-optimal power profiles, different waveforms, under different modulation $\Uw$, have different optimal data correlation $\Rxw^\star=\Uw^\herm\Rs^\star\Uw$.
			\item Under isotropic transmission (uniform TF symbol-level power profile), e.g., $|[\s]_{q'}|^2={\pt}$ for $ q'\in \mathcal{Q}$ or $\vdiag\{\Rsw\}=\frac{\pt}{Q}\mathbf{1}_Q$ \cite{OFDM+OTFS,AFDM_ISAC}, different waveforms exhibit the same CRB performance.
		\end{itemize}
		\item[ii)] For a given data sequence $\x$ (hence a \emph{given} $\Rx$), different waveforms generally lead to different TF symbol-level power profile $\vdiag\{\Rsw\}$, due to different projection characteristic $\Uw$. Consequently, their DD estimation capabilities differ. Hence, the waveform sensing advantage is data-dependent. For waveforms that happen to project this \emph{particular} $\Rx$ to be more aligned with the FIM structure, they achieve better CRB performance. After changing the data sequence $\x$ to a new realization, the advantage may disappear or even reverse.
	\end{itemize}
\end{remark}

\begin{proposition}[CRB for single-target DD estimation]
	\label{proposition:CRB_single_target}
	Under single-target scenarios, sinc pulse shaping, and isotropic transmission (e.g., $\vdiag\{\Rsw\}= \frac{\pt}Q\mathbf{1}_Q$ or each FT symbol $[\s]_{q'}$ has the same envelope of $\sqrt{\pt}$), the estimation of $\tau$ and $\nu$ is decoupled. The CRBs for $\tau$ and $\nu$ are respectively	
	\begin{align}
		\varepsilon(\tau) &= \frac{3M^2\Ts^2}{2|g|^2\pi^2 Q(M^2-1)\SNR} \simeq \frac{3\Ts^2}{2|g|^2\pi^2MN\SNR};  \label{eqn:CRB_single_target} \\
		\varepsilon(\nu) &= \frac{3}{2|g|^2\pi^2Q(N^2-1)\Tsym^2\SNR}\simeq\frac{3}{2|g|^2\pi^2MN^3\Tsym^2\SNR}, \notag
	\end{align}
	where $\SNR:=\pt/\noisevar$ is the (per-symbol) transmit SNR.

	\begin{IEEEproof}
		See Appendix \ref{appendix:CRB_single_target} in Supplementary Information.
	\end{IEEEproof}
\end{proposition}

\begin{remark}
	The expressions \eqref{eqn:CRB_single_target} provide an intuitive understanding of how key system parameters affect joint DD estimation performance. These expressions are waveform-independent under the isotropic transmission assumption, coinciding with the well-established results for OFDM in \cite{OFDM+OTFS,braun2014ofdm}, and AFDM in \cite{AFDM_ISAC}. Note that these CRBs are not the ``optimal'' ones, as the power distribution among TF symbols can be further designed (see next section).
\end{remark}

\section{Fundamental Joint Delay-Doppler Estimation-Theoretic Performance Limits}
\label{sec:Pareto}

\subsection{Problem Formulation}

By varying the TF symbol power profile $\vdiag\{\Rs\}$, the FIM, and hence, the CRBs for delay $\varepsilon(\boldtau)$ and Doppler $\varepsilon(\boldnu)$, change accordingly. Under the power constraint \eqref{eqn:power_constraint}, reducing one CRB necessarily increases the other. One core interest in academia is to investigate the optimal trade-off between the two CRBs, which is captured by the Pareto front of the achievable DD CRB region. 

In the sequel, this paper formulates and solves the related Pareto optimization problem to explore the estimation-theoretic limits. In particular, this paper focuses on finding the transmit waveform (more specifically, the TF symbol power profile $\vdiag\{\Rs\}$ or the correlation matrix $\Rxw$) to optimize the CRBs, based upon given parameter $\boldtheta$ \cite{range_compression_TSP,MIMO_radar_signal_design_TSP,TWC_Hua}.\footnote{As seen in Theorem \ref{theorem:FIM_structure}, the FIM and CRB matrices inherently depend on the true (and unknown) parameter $\boldtheta$. Practically, to assess the performance, the FIM/CRB can be evaluated at a specific parameter operating point \cite{MIMO_radar_signal_design_TSP,range_compression_TSP}, e.g., a particular DD bin of interest. The parameter values can be nominal values \cite{MIMO_radar_signal_design_TSP} or estimated from an initial probing stage \cite{range_compression_TSP}. Following common practice in theoretical studies \cite{range_compression_TSP,MIMO_radar_signal_design_TSP,TWC_Hua}, this paper does not address how these parameters are explicitly obtained, as numerous estimation techniques have already been well-discussed, e.g., in \cite{braun2014ofdm,multi-tone,OFDM+OTFS2}.} Before proceeding, it should be noted that: i) The elements of the FIM differ by several orders of magnitude, making it ill-conditioned and causing numerical difficulties; ii) The diagonal entries of the CRB matrix have different physical units, so conventional optimization objectives, e.g., minimizing the trace of the CRB matrix \cite{TWC_Hua}, are not always meaningful. To tackle these issues, define the transformed variable
\begin{equation*}
	\begin{aligned}
		\tilde{\boldtau} &:= \boldtau/\Tsym; & \tilde{\boldnu} &:= \boldnu\times \Tsym; \\
		\tilde{\g} &:= \g; & \tilde{\boldtheta} &:= \Big[\tilde{\boldtau}^\trans,\tilde{\boldnu}^\trans,\tilde{\g}^\trans\Big]^\trans.
	\end{aligned}
\end{equation*}
Thus, $\boldtheta$ and $\tilde{\boldtheta}$ relate with each other according to $\tilde\boldtheta = \D{\boldtheta}$, with $\D:=\Mdiag\big\{\frac1\Tsym\I_P,\Tsym\I_P,\I_{2P}\big\}$. Under such transformation, each element in $\tilde{\boldtheta}$ has the same unit and comparable magnitude. It is straightforward to verify that \cite[Chapter 3]{Kay1993Fundamentals}
\begin{equation*}
	\FIM_{\tilde{\boldtheta}\tilde{\boldtheta}} = \D^{-1}\FIM_{\boldtheta\boldtheta}\D^{-1};\quad \CRB_{\tilde{\boldtheta}} = \D\CRB_{\boldtheta} \D.
\end{equation*}
The Pareto front of the joint DD CRB region is characterized by optimizing a weighted sum of the delay and Doppler CRBs:
\begin{equation}
	\begin{aligned}
		\min_{\Rxw}& & & \alpha \Tr(\CRB_{\tilde\boldtau}) + (1-\alpha)\Tr(\CRB_{\tilde\boldnu}) \\
		\st& & &\Tr(\Rxw)=\pt, \\
		&\ & & \Rxw = \Rxw^\herm \succeq \mathbf{0}, \\
		&\ & & \Rank(\Rxw) = 1,
	\end{aligned}
	\label{eqn:CRB_optimization_original}
\end{equation}
for each $\mathsf{w}\in\mathcal{W}$, where $\alpha\in[0,1]$ is the relative weight. By solving the \eqref{eqn:CRB_optimization_original} for each $\alpha\in[0,1]$, it traces the complete Pareto-optimal boundary. 

\begin{remark}
	The above optimization problem investigates the fundamental limits of DD estimation under finite power budgets. In ISAC systems, it corresponds to the sensing-optimal operating point, which prioritizes the maximization of sensing performance without imposing any communication constraints \cite{wang_2026_TIT}. For practical ISAC implementations, waveform design is further bounded by communication-related constraints including peak-to-average power ratio, modulation order, and bit error rate. These coupled design requirements remain open topics for future investigation.
\end{remark}

Problem \eqref{eqn:CRB_optimization_original} is a high-dimensional rank-constrained optimization, and is inherently challenging to solve. In this regard, the following proposition is proposed to equivalently transform the original optimization into a more tractable form.

\begin{proposition} 
	\label{Proposition:optimization_transform}
	The original problem \eqref{eqn:CRB_optimization_original} is equivalent to
	\begin{equation}
		\label{eqn:CRB_optimization_reduced}
		\begin{aligned}
			\min_{\boldrho}& & &\Tr[\A(\alpha)\FIM_{\boldtheta\boldtheta}^{-1}(\boldrho)]\\
			\st& & & \mathbf{1}_Q^\trans\boldrho=\pt,\ \boldrho\succeq \mathbf{0}_Q,
		\end{aligned}
	\end{equation}
	where $\A(\alpha):={\Mdiag}\big\{\frac\alpha{\Tsym^2}\I_P,(1-\alpha)\Tsym^2\I_P,\mathbf{0}_{2P}\big\}$, $\boldrho\in\mathbb{R}^Q$, and $\FIM_{\boldtheta\boldtheta}(\boldrho)$ is the 
	FIM of $\boldtheta$ in Theorem \ref{theorem:FIM_structure}, as a function of $\boldrho=[\varrho_1,\cdots,\varrho_Q]^\trans$, written as
	\begin{equation}
		\label{eqn:FIM_decompose}
		\FIM_{\boldtheta\boldtheta}(\boldrho)=\sum_{q'\in\mathcal{Q}} \varrho_{q'} \H_{q'} \quad {\text {with}} \quad \H_{q'}:=
		\FIM_{\boldtheta\boldtheta}\bigg|_{\Rs=\mathbf{e}_{q'}\mathbf{e}_{q'}^\trans}.
	\end{equation}
	
	\begin{IEEEproof}
		See Appendix \ref{appendix:optimization_transform} in Supplementary Information.
	\end{IEEEproof}
\end{proposition}


\begin{remark}
	With proposition \ref{Proposition:optimization_transform}, the dimensionality of the optimization variables is greatly reduced, and the optimization problem \eqref{eqn:CRB_optimization_reduced} is convex. Consequently, the computational complexity for CRB-minimization is significantly reduced. 
\end{remark}

Once the Pareto-optimal solution, denoted as $\boldrho^\star(\alpha)=[\varrho_1^\star(\alpha),\cdots,\varrho_Q^\star(\alpha)]^\trans$, to \eqref{eqn:CRB_optimization_reduced} is obtained, the optimal transmit TF-domain symbol is constructed as 
$$\s^{\star}(\alpha) = \sqrt{Q} \Big[\sqrt{\varrho_1^\star(\alpha)}\mathrm{e}^{\jj\varphi_1},\cdots,\sqrt{\varrho_Q^\star(\alpha)}\mathrm{e}^{\jj\varphi_Q}\Big]^\trans,$$
where $\varphi_1,\cdots,\varphi_Q\in\mathbb{R}$ are arbitrary. The optimal TF symbol correlation matrix $\Rs^\star(\alpha)$ is thus $\Rs^\star=\frac1Q\s^\star(\alpha)[\s^\star(\alpha)]^\herm$. This shows that the CRB-optimal sequence has some free degrees of freedom, i.e., the phase terms $\{\varphi_{q'}\}_{q'\in\mathcal{Q}}$, that can be arbitrarily chosen without sacrificing any optimality in terms of the CRB. In particular, this naturally suggests that (Q)PSK symbols with optimal power allocation can achieve non-trivial communication rate while being DD estimation-optimal, e.g., for ISAC applications. Furthermore, the waveform-specific CRB-optimal data sequence has the following general form of $\xw^\star(\alpha) = \Uw^\herm\s^{\star}(\alpha)$. 

Given $\Rs^\star(\alpha)$, the CRBs for delay and Doppler can be calculated from evaluating the FIM by replacing $\Rs$ with $\Rs^\star(\alpha)$ and invoking the FIM-CRB relationship in Definition \ref{definition:CRB}. Denote the CRBs for $\boldtau$ and $\boldnu$ in \eqref{eqn:CRB_tau/nu} under the optimal TF symbol power profile $\boldrho^\star(\alpha)$ as $\varepsilon(\boldtau,\alpha)$ and $\varepsilon(\boldnu,\alpha)$. By sweeping $\alpha$ from $0$ to $1$, the Pareto front of the joint range-Doppler CRB region can be readily obtained as
\begin{equation*}
	\Big\{\big(\varepsilon(\boldtau,\alpha),\varepsilon(\boldnu,\alpha)\big)\Big|\alpha\in[0,1]\Big\}.
\end{equation*}
In the sequel, a generic approach to obtain the numerical solution is provided.

\subsection{Identifying the Delay-Doppler CRB Pareto Boundary}

\begin{algorithm}[t]
\caption{}
\label{alg}
\begin{algorithmic}[1]
\Require FIM decomposition $\{\mathbf H_{q'}\}_{q'=1}^Q$, weighting matrix $\A(\alpha)$, total power $\pt$
\Ensure Optimal power $\boldrho^\star(\alpha)=[\varrho_1^{[k]}(\alpha),\cdots,\varrho_{Q}^{[k]}(\alpha)]^\trans$

\State Initialize $\varrho_{q'}^{[0]}(\alpha) = \pt / Q$, $\forall q'\in\mathcal{Q}$
\State Set iteration index $k=0$

\Repeat
    \State Compute $\FIM_{\boldtheta\boldtheta}^{[k]}:= \FIM_{\boldtheta\boldtheta}(\boldrho^{[k]}(\alpha)) = \sum_{q'\in\mathcal{Q}} \varrho_{q'}^{[k]}(\alpha) \mathbf H_{q'}$
    \For{$q' = 1,\ldots,Q$}
        \State Compute marginal utility         
        \[
        \kappa_{q'}^{[k]}(\alpha) =
        \Tr\!\left[
        \big(\FIM_{\boldtheta\boldtheta}^{[k]}\big)^{-1}
        \A(\alpha)
        \big(\FIM_{\boldtheta\boldtheta}^{[k]}\big)^{-1}
        \mathbf H_{q'}
        \right]
        \]
    
    \State Update power allocation
    \[
    \varrho_{q'}^{[k+1]}(\alpha)  = \max\left\{0,
    \varrho_{q'}^{[k]}(\alpha)  + \vartheta \kappa_{q'}^{[k]}(\alpha)\right\}
    \]
    	\textcolor{magenta!70}{\Comment{\textit{$\vartheta$ is the pre-set step size for parameter update}}}
    \EndFor
    \State Normalize $\sum_{q'=1}^Q \varrho_{q'}^{[k+1]}(\alpha)  \gets \pt$
    \State $k \gets k+1$
\Until{convergence}
\end{algorithmic}
\end{algorithm}

Since $\FIM_{\boldtheta\boldtheta}(\boldrho)$ is affine in $\boldrho$ and $\Tr[\A(\alpha)\FIM_{\boldtheta\boldtheta}^{-1}(\boldrho)]$ is convex over the cone of positive definite matrices, problem \eqref{eqn:CRB_optimization_reduced} is a convex optimization problem. Constructing the Lagrangian as
\begin{equation*}
	\mathcal{L}(\alpha)=\Tr[\A(\alpha)\FIM_{\boldtheta\boldtheta}^{-1}(\boldrho)]+\lambda\Big(\sum_{q'}\varrho_{q'}-\pt\Big)-\sum_{q'}\mu_{q'}\varrho_{q'},
\end{equation*}
where $\lambda$ and $\{\mu_{q'}\}_{q'\in\mathcal{Q}}$ are associated Lagrange multipliers. The first-order stationary condition is
\begin{equation*}
	\pd{}{\varrho_{q'}} \mathcal{L}(\alpha) = -\Tr[\FIM_{\boldtheta\boldtheta}^{-1}(\boldrho)\A(\alpha)\FIM_{\boldtheta\boldtheta}^{-1}(\boldrho)\H_{q'}]+\lambda-\mu_{q'}=0.
\end{equation*}
By defining $\kappa_{q'}(\alpha) := \Tr[\FIM_{\boldtheta\boldtheta}^{-1}(\boldrho)\A(\alpha)\FIM_{\boldtheta\boldtheta}^{-1}(\boldrho)\H_{q'}]$ for simplicity and invoking the complementary slackness condition $\mu_{q'} \varrho_{q'}=0$, the stationarity condition yields the following Karush-Kuhn-Tucker (KKT) condition
\begin{equation*}
	\begin{cases}
		\kappa_{q'}(\alpha) = \lambda & \text{ for } \varrho_{q'}^\star(\alpha) > 0,\\
		\kappa_{q'}(\alpha) \leq \lambda & \text{ for } \varrho_{q'}^\star(\alpha) = 0,
	\end{cases}
\end{equation*}
where $\varrho_{q'}^\star(\alpha)$ is the optimal solution of $\varrho_{q'}$ under $\A(\alpha)$. Since problem \eqref{eqn:CRB_optimization_reduced} is convex, any stationary point satisfying the above KKT condition is globally optimal. By exploiting the idea of gradient descent, a numerical solution is given in Algorithm \ref{alg}, which converges to the global optimal solution for a proper choice of the step size $\vartheta$.

\begin{remark}[Complexity Analysis] From \eqref{eqn:J}, each $\H_{q'}$ requires $\mathcal{O}\!\left(6P^2\times Q\right)$ complexity; the overall complexity in computing $\{\H_{q'}\}_{q'\in\mathcal{Q}}$ is therefore $\mathcal{O}(Q^2 P^2)$. At each iteration in Algorithm \ref{alg}, the FIM $\FIM_{\boldtheta\boldtheta}^{[k]} = \sum_{q'} \varrho_{q'}^{[k]}(\alpha) \mathbf H_{q'}$ is constructed with $\mathcal O(Q P^2)$ complexity, followed by matrix inversion with $\mathcal O(P^3)$ complexity. The marginal utilities $\kappa_{q'}^{[k]}(\alpha)$ require $\mathcal O(P^3)$ per $q'$, resulting in $\mathcal O(Q P^3)$ complexity for $Q$ evaluations. The non-negativity update of $\varrho^{[k+1]}_{q'}(\alpha)$ has $\mathcal O(Q)$ complexity. 
Thus, the overall per-iteration complexity is $\mathcal O(Q P^3)$. The total complexity is $\mathcal O(K_{\mathsf{iter}}Q P^3+Q^2P^2)$, where $K_{\mathsf{iter}}$ is the number of iterations. 
\end{remark}

\begin{table}[t]
	\centering
	\caption{Default Parameter Setting}
	\label{table:parameter}
	\begin{tabular}{lcl}
		\toprule
		Parameter & Symbol & Value/Range \\
		\midrule
		\multicolumn{3}{l}{
			{\bf IEEE 802.11p} \cite{OFDM+OTFS,Delay_Doppler_processing_numerical1,Delay_Doppler_processing_numerical2}} \\
		Bandwidth & $B$ & \SI{10}{MHz}  \\
		Center Frequency & $\fc$ & $\SI{5.89}{GHz}$ \\
		Fast-Time Samples & $M$ & $64$ \\
		Prefix (CP/CPP) Length & $\Mcp$ & $M/4$ \\
		Slow-Time Samples & $N$ & $50$ \\
		\midrule
		Number of Targets & $P$ & $4$ \cite{wang2025device,OTFS_CRB_uplink,OFDM+OTFS2,entropy_OTFS} \\
		Transmit SNR & $\SNR={\pt}/{\noisevar}$ & \SI{10}{dB} \\
		\bottomrule
	\end{tabular}
\end{table}
\section{Numerical Examples}
\label{sec:numerical}

\begin{figure}[t]
	\centering
	\includegraphics[scale=1]{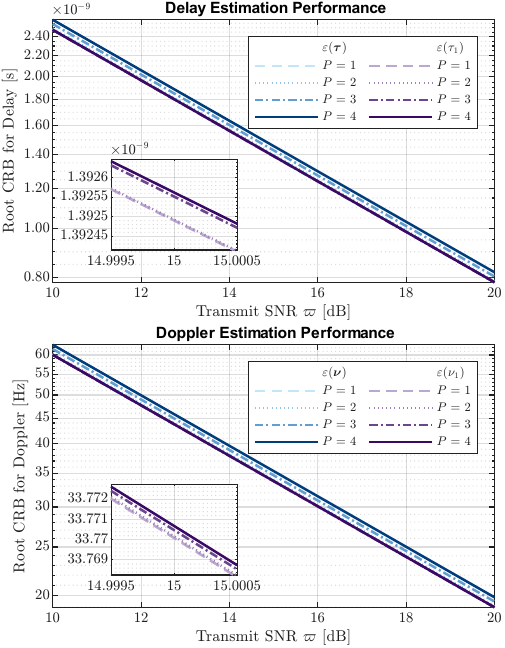}
	\caption{Delay-Doppler CRBs (under isotropic transmission $|[\s]_{q'}|=\sqrt{\pt}$) vs. transmit SNR $\SNR$ for various the number of targets $P$.}
	\label{fig:CRB_SNR}
\end{figure}

\begin{figure}[t]
\centering
	\includegraphics[scale=1]{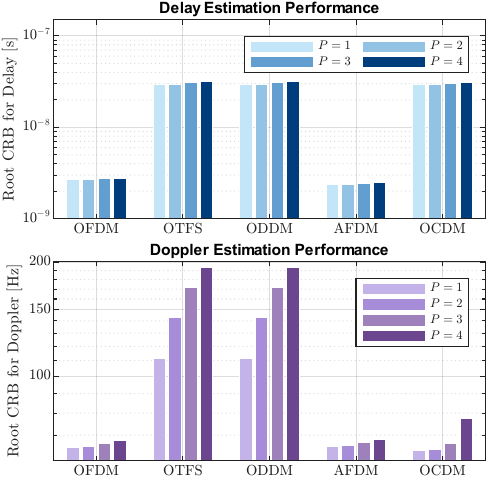}
	\caption{Joint delay-Doppler CRBs for different waveforms, under the given particular realization of data sequence \eqref{eqn:data_correlation_example} as an example.}
	\label{fig:CRB_data_correlation}
\end{figure}

\subsection{Setups}
The general setting of this part is as follows. The waveform-related parameters are based on IEEE 802.11p, which has not only been standardized for dedicated short-range communications \cite{Delay_Doppler_processing_numerical2} but also demonstrated potentials for radar sensing \cite{Delay_Doppler_processing_numerical1,OFDM+OTFS}. For the channel with $P=4$, the delays, Dopplers, and channel coefficients are set as \cite{entropy_OTFS}
\begin{equation*}
	\begin{aligned}
		\boldtau &= [0, 8, 4, 6]^\trans \Tsym/M; \\
		\boldnu &= [4.82,-3.23, 1.38,-2.47]^\trans /(N\Tsym); \\
		\g &= [-2-9\jj, 40 + 73\jj, 3 + 45\jj, 15-43\jj]^\trans/100.
	\end{aligned}
\end{equation*}
Additionally, the rolling parameter for RRC pulse is $\beta=0.1$ \cite{ODDM_overall}; the expansion parameter for the Gaussian pulse is $\gamma=2$; the chirp parameters for AFDM are set as $c_1=9/512$, $c_2=0$ \cite{AFDM_ISAC_MIMO}. Unless otherwise specified, all default parameters are summarized in TABLE \ref{table:parameter}. All results herein are based on the derived theoretical findings and obtained using MATLAB.

\subsection{Results}
In Fig. \ref{fig:CRB_SNR}, the effects of transmit SNR and the number of sensing targets on DD CRBs are portrayed. In this example, RRC pulse is utilized, and isotropic transmission (uniform power distribution on TF symbols) is considered, so that the DD CRBs are the same for all waveforms. It can be seen that root CRBs on both delay and Doppler in $\log$-scale decrease linearly with respect to SNR in $\rm{dB}$, as expected. Moreover, increasing the number of targets degrades both per-target accuracy $\varepsilon(\nu_1)$, $\varepsilon(\tau_1)$ and overall estimation accuracy $\varepsilon(\boldnu)$, $\varepsilon(\boldtau)$ (i.e., increases the corresponding CRBs). This degradation stems from an expanded parameter-under-estimation space, which introduces more unknowns and intensifies inter-target interference in the observation model, ultimately reducing the Fisher information per parameter. 

To examine the effects of waveform modulation on DD CRBs, the following data sequence, as an example with the RRC pulse, is considered for all waveforms
\begin{equation}
	\begin{aligned}
		\x \leftarrow [1,\cdots&,N]^\trans \otimes \mathbf{1}_{M\times1}
		+\mathbf{1}_{N\times1}\otimes [1,\cdots,M]^\trans,
	\end{aligned}
	\label{eqn:data_correlation_example}
\end{equation}
which is then normalized according to the constraint \eqref{eqn:power_constraint}. Based on \eqref{eqn:data_correlation_example}, the TF-domain symbol correlation matrix $\Rsw$ can be calculated for each $\mathsf{w}\in\mathcal{W}$, and corresponding DD CRBs are further obtained in Fig. \ref{fig:CRB_data_correlation}. It is crucial to note that the presented results do \emph{not} imply that certain waveforms are fundamentally superior to others for estimation. Instead, they show that, for the considered particular data \eqref{eqn:data_correlation_example}, the modulation matrix $\Uw$ determines how well the resulting TF symbol-level power distribution aligns with the structure of the FIM, which in turn governs the estimation performance. Also note that OTFS and ODDM yield identical CRBs here, as $\U_{\OTFS}=\U_{\ODDM}$ in the considered setup. In a practical system, the CRB difference, if any, between these two waveforms would stem from other implementation-specific aspects, not from the modulation $\Uw$ itself.

\begin{figure}[t]
	\centering
	\includegraphics{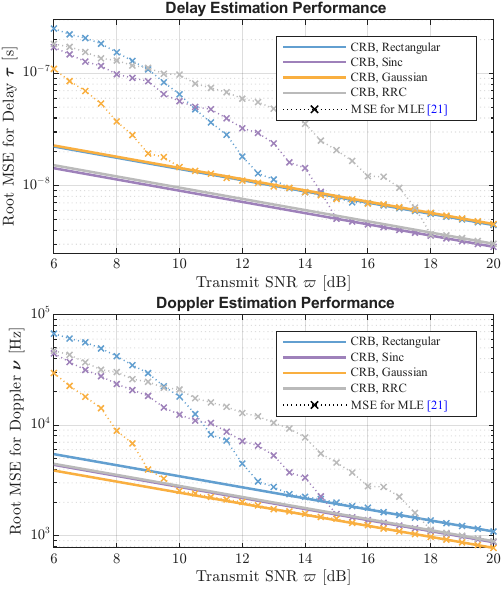}
	\caption{CRB and MSE performance for joint delay-Doppler estimation for various pulse shape $\mathsf{ps}$.}
	\label{fig:MLE}
\end{figure}

In Fig. \ref{fig:MLE}, the CRB and MSE for DD estimation versus transmit SNR are presented. Specifically, the MSE curves are obtained by directly applying the technique developed in \cite[Section IV]{OFDM+OTFS2}. It iteratively alternates optimizing multipath DD parameters and updating channel gains via low-complexity two-stage search. By neglecting cross terms in the log-likelihood loss assuming sparse multipath channels, it achieves \emph{approximately} maximum likelihood estimation (MLE) \cite{OFDM+OTFS2}. In this specific example, $M=N=16$ and isotropic power allocation is considered. It is observed that different pulse shapes yield different DD CRBs, and their respective threshold SNRs also differ. For all pulse shapes, it is remarkable to note that the derived CRB is able to accurately predict the estimator performance in moderate-to-high SNR regime.

\begin{figure}[t]
	\centering
	\includegraphics{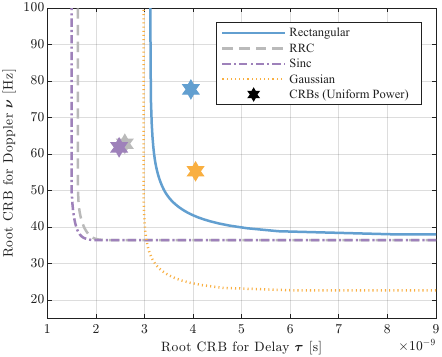}
	\caption{The Pareto-optimal and uniform-power delay-Doppler CRBs for various pulse shape $\mathsf{ps}$. Both the Pareto-optimal and isotropic CRBs are identical across various waveforms.}
	\label{fig:Pareto}
\end{figure}

In Fig. \ref{fig:Pareto}, the joint DD CRB region is characterized, wherein both the CRB Pareto front and the DD CRBs under isotropic transmission are plotted for various pulse shapes. For Algorithm \ref{alg}, it is regarded as convergent when the relative change of CRBs in the last iteration are within $ 10^{-6}$, and the relative change of gradient satisfies
$$\left|{\sqrt{\textstyle\sum_{q'}\left[\kappa_{q'}^{[k]}(\alpha)\right]^2}}\Bigg/{\sqrt{\textstyle\sum_{q'}\left[\kappa_{q'}^{[k-1]}(\alpha)\right]^2}}-1\right|\leq 10^{-4}.$$
It is observed that for all pulse shapes, the Pareto front exhibits a distinctly convex shape, with the DD CRBs from isotropic transmission lying considerably far from this optimal boundary. This clearly reveals that the system's ultimate estimation-theoretic capability can far exceed the CRBs under isotropic transmission --- a scenario that is implicitly adopted as the default sensing benchmark in many existing studies \cite{OFDM+OTFS,braun2014ofdm}. 

Following the results in Fig. \ref{fig:Pareto}, the optimal TF symbol power profile $\boldrho^\star(\alpha)$ (normalized) is plotted in Fig. \ref{fig:Pareto_signal} for various $\alpha$ and pulse shapes. These sparse patterns reveal an intuitive structure tied to estimation preference. For $\alpha=0$ (minimizing the Doppler CRB), the power is allocated to TF samples at the extremes of the slow-time index; for $\alpha=1$ (minimizing the delay CRB), power is instead concentrated at some of the range-frequency index (i.e., the discrete Fourier transformed fast-time index) that are jointly determined by the index $|\eta_m|$ and weighting $\Upsilon_{\mathsf{ps}}(m)$. 
As $\alpha$ sweeps from $0$ to $1$, the optimal pattern is observed to transit smoothly in between. The power distribution behavior is physically justified, as Doppler estimation couples primarily with slow-time variations, whereas delay estimation couples with fast-time sampling. When the pulse shaping spectrum is flat, i.e., $\mathsf{ps}=\mathsf{sinc}$, power is distributed around $m=M/2$, where $|\eta_m|$ is maximized. For other pulses, e.g., $\mathsf{ps}\in\{\mathsf{rect},\mathsf{rrc},\mathsf{Gaussian}\}$, their pulse shaping functions $\Upsilon_{\mathsf{ps}}(m)$, which is minimized around $m=M/2$, effectively suppresses the Fisher information contribution from TF samples around index $m=M/2$ and thus shifts the optimal power distribution away from $m=M/2$.

\begin{figure*}[t]
	\centering
	\includegraphics{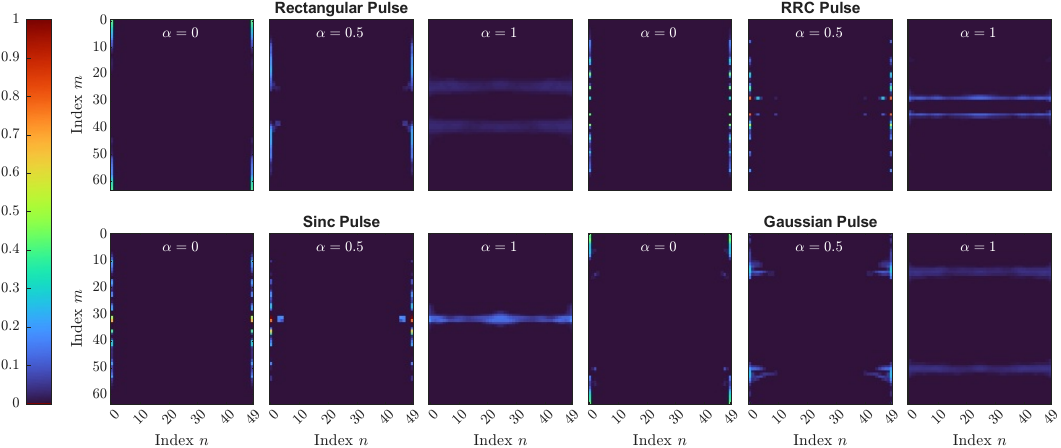}
	\caption{The normalized Pareto-optimal TF-domain symbol-level power profile, i.e., normalized $\boldrho^\star(\alpha)$, with various DD preference $\alpha$ and pulse shape $\mathsf{ps}$.}
	\label{fig:Pareto_signal}
\end{figure*}

\section{Conclusion}

In this paper, a systematic and unified Fisher information framework was established to quantify CRB performance on joint DD estimation using various waveforms. Within this framework, the effects of channel dynamics, waveform-specific data-to-signal modulation, the number of targets, etc., on the performance of joint DD estimation were examined. The Pareto front on the DD CRB region was identified by formulating and solving a family of power allocation problems. The resulting front serves as the system's ultimate DD estimation-centric capability. Based upon these analytical results, this paper serves as a building block to accommodate more waveform-specific characteristics for future study, and hopefully sheds light on the design of sensing/ISAC-oriented advanced waveforms in next-generation networks.
\appendices
\numberwithin{equation}{section}

\section{Proof of Theorem \ref{theorem:FIM_structure}}
\label{Appendix:FIM_structure_proof}

\emph{Proof:} The overall FIM $\FIM_{\boldtheta\boldtheta}$ can be block-partitioned as
\begin{equation*}
	\label{eqn:FIM_partition}
	\renewcommand{\arraystretch}{1.1}
	\FIM_{\boldtheta\boldtheta} = \begin{bmatrix}
		\FIM_{\boldtau\boldtau} & \FIM_{\boldtau\boldnu} & \FIM_{\boldtau\gr} & \FIM_{\boldtau\gi} \\
		\FIM^\trans_{\boldtau\boldnu} & \FIM_{\boldnu\boldnu} & \FIM_{\boldnu\gr} & \FIM_{\boldnu\gi} \\
		\FIM^\trans_{\boldtau\gr} & \FIM^\trans_{\boldnu\gr} & \FIM_{\gr\gr} & \FIM_{\gr\gi}  \\
		\FIM^\trans_{\boldtau\gi} & \FIM^\trans_{\boldnu\gi} & \FIM^\trans_{\gr\gi} & \FIM_{\gi\gi}
	\end{bmatrix}.
\end{equation*}
In the sequel, each $P\times P$ sub-block is calculated. Before proceeding, the following intermediate derivatives are calculated
\begin{equation*}
	\begin{aligned}
		\pd{\r(\boldtheta)}{\tau_p} &= \NU\G (\e_p\e_p^\trans\otimes\I_Q)\dot{\bm{\mathcal{T}}}_{\!\!\mathsf{ps}}
		\s; \\
		\pd{\r(\boldtheta)}{\nu_p} &= \dot{\NU}(\e_p\e_p^\trans \otimes \I_Q)\G\TAU\s;\\
		\pd{\r(\boldtheta)}{\Re\{g_p\}} &= \NU(\e_p\e_p^\trans\otimes\I_Q)\TAU\s; \\
		\pd{\r(\boldtheta)}{\Im\{g_p\}} &= \jj\NU(\e_p\e_p^\trans\otimes\I_Q)\TAU\s=\jj \pd{\r(\boldtheta)}{\Re\{g_p\}},
	\end{aligned}
\end{equation*}
where  $\e_p$ denotes the $p$-th standard orthonormal basis, i.e., the $p$-th column of $\I_P$.

\begin{figure*}[hb]
	\rule{\linewidth}{0.2pt}
	\begin{align}
		& \pd{\r^\herm(\boldtheta)}{\nu_{p_1}}\pd{\r(\boldtheta)}{\nu_{p_2}} =\s^\herm\TAU^\herm\G^* (\e_{p_1}\e_{p_1}^\trans \otimes \I_Q)\dot{\NU}^\herm\dot{\NU}(\e_{p_2}\e_{p_2}^\trans \otimes \I_Q)	\G\TAU\s \label{eqn:J_nu_calculation}
		\\
		&= \Tr\Big[\!\!\! \underbrace{ (\e_{p_1}^\trans\otimes\I_Q)\dot{\NU}^\herm\dot{\NU}(\e_{p_2}\otimes\I_Q)}_{\text{The }(p_1,p_2)\text{-th } Q\times Q \text{ sub-block of }\dot{\NU}^\herm\dot{\NU}}	\!\!\!\!
		\overbrace{(\e_{p_2}^\trans\otimes\I_Q)\G\TAU\s\s^\herm\TAU^\herm\G^*(\e_{p_1}\otimes \I_Q)}^{\text{The }(p_1,p_2)\text{-th } Q\times Q \text{ sub-block of }Q\G^*\TAU^*\Rs^\trans\TAU^\trans\G}\!\!\!\Big]
		= Q\Tr\Big[ \Big(\dot{\NU}^\herm\dot{\NU}\Bodot\G^*\TAU^*\Rs^\trans\TAU^\trans\G\Big)_{\{p_1,p_2\}} \Big]\notag
	\end{align}
\end{figure*}

\vspace{0.3cm}
\noindent \emph{i) Calculation of FIM blocks $\FIM_{\boldnu\boldnu}$, $\FIM_{\boldnu\boldtau}$, and $\FIM_{\boldtau\boldtau}$}

From the Slepian-Bangs formula (see \cite[Theorem 2.1]{Complex_Stochastic_Processes}), the $(p_1,p_2)$-th element in $\FIM_{\boldnu\boldnu}$ is
\begin{equation*}
	[\FIM_{\boldnu\boldnu}]_{p_1,p_2} = \frac{2}{\noisevar}\Re\Big\{ \pd{\r^\herm(\boldtheta)}{\nu_{p_1}}\pd{\r(\boldtheta)}{\nu_{p_2}} \Big\},
\end{equation*}
where the quantity inside $\Re\{\cdot\}$ is given in \eqref{eqn:J_nu_calculation}. It can be further observed that \eqref{eqn:J_nu_calculation} is the trace of the $(p_1,p_2)$-th sub-block of $\dot{\NU}^\herm\dot{\NU}\Bodot\G^*\TAU^*\Rs^\trans\TAU^\trans\G^\trans$. Hence, $\FIM_{\boldnu\boldnu}$ reduces to
\begin{equation*}
	\frac{2}{\noisevar}\Re\Big\{\underbrace{Q\BTr\Big[\dot{\NU}^\herm\dot{\NU} \Bodot \G^*\TAU^*\Rs^\trans \TAU^\trans\G\Big]}_{:=\J_{\boldnu\boldnu}}\Big\},
\end{equation*}
yielding $\J_{\boldnu\boldnu}$ in \eqref{eqn:J_nu}. Similarly, the FIM $\FIM_{\boldtau\boldnu}$ is related to $\J_{\boldtau\boldnu}$, with its $(p_1,p_2)$-th entry given as
\begin{equation*}
	\begin{aligned}
		&\ Q\Tr\Big[ \Big(\G^*\NU^\herm\dot{\NU} \Bodot \dot{\bm{\mathcal{T}}}_{\!\!\mathsf{ps}}^*\Rs^\trans \TAU^\trans\G\Big)_{\{p_1,p_2\}} \Big] \\
		=&\ Q\Tr\Big[ (\G^*\NU^\herm\dot{\NU})_{\{p_1,p_2\}}(\dot{\bm{\mathcal{T}}}_{\!\!\mathsf{ps}}^*\Rs^\trans \TAU^\trans\G)_{\{p_1,p_2\}}\Big]\\
		=&\ Q\Tr\Bigg[ \Big(\sum_{k=1}^P(\G^*)_{\{p_1,k\}}(\NU^\herm\dot{\NU})_{\{k,p_2\}}\Big)\\
		&\hspace{2.5cm} \Big(\sum_{\ell=1}^P(\dot{\bm{\mathcal{T}}}_{\!\!\mathsf{ps}}^*\Rs^\trans \TAU^\trans)_{\{p_1,\ell\}}\G_{\{\ell,p_2\}}\Big)\Bigg]\\
		=&\ Q\Tr\Big[ (\G^*)_{\{p_1,p_1\}}(\NU^\herm\dot{\NU})_{\{p_1,p_2\}}(\dot{\bm{\mathcal{T}}}_{\!\!\mathsf{ps}}^*\Rs^\trans \TAU^\trans)_{\{p_1,p_2\}}\G_{\{p_2,p_2\}}\Big]\\
		=&\ Q\Tr\Big[ (\NU^\herm\dot{\NU})_{\{p_1,p_2\}}(\G^*)_{\{p_1,p_1\}}(\dot{\bm{\mathcal{T}}}_{\!\!\mathsf{ps}}^*\Rs^\trans \TAU^\trans)_{\{p_1,p_2\}}\G_{\{p_2,p_2\}}\Big]\\
		=&\ Q\Tr\Big[ (\NU^\herm\dot{\NU})_{\{p_1,p_2\}}(\G^*\dot{\bm{\mathcal{T}}}_{\!\!\mathsf{ps}}^*\Rs^\trans \TAU^\trans\G)_{\{p_1,p_2\}}\Big]\\
		=&\ Q\Tr\Big[ \Big(\NU^\herm\dot{\NU} \Bodot \G^*\dot{\bm{\mathcal{T}}}_{\!\!\mathsf{ps}}^*\Rs^\trans \TAU^\trans\G\Big)_{\{p_1,p_2\}} \Big],
	\end{aligned}
\end{equation*}
where the fact has been used that $\G$ is diagonal with $\G_{\{p,p\}}=g_p\I_Q$ and $\G_{\{p,k\}}=\mathbf{0}_{Q\times Q}$ for $p\neq k$; the FIM $\FIM_{\boldtau\boldtau}$ is related to $\J_{\boldtau\boldtau}$, with its $(p_1,p_2)$-th entry given as
\begin{equation*}
	\begin{aligned}
		[\J_{\boldtau\boldtau}]_{p_1,p_2} &= Q\Tr\Big[\Big(\G^*\NU^\herm\NU\G \Bodot\dot{\bm{\mathcal{T}}}_{\!\!\mathsf{ps}}^*\Rs^\trans\dot{\bm{\mathcal{T}}}_{\!\!\mathsf{ps}}^\trans\Big)_{\{p_1,p_2\}}\Big]\\
		&= Q\Tr\Big[\Big(\NU^\herm\NU\Bodot\G^*\dot{\bm{\mathcal{T}}}_{\!\!\mathsf{ps}}^*\Rs^\trans\dot{\bm{\mathcal{T}}}_{\!\!\mathsf{ps}}^\trans\G\Big)_{\{p_1,p_2\}}\Big].\\
	\end{aligned}
\end{equation*}
Hence, \eqref{eqn:J_tau} and \eqref{eqn:J_taunu} can be obtained.

\vspace{0.3cm}
\noindent \emph{ii) Calculation of FIM blocks $\FIM_{\gr\gr}$, $\FIM_{\gr\gi}$, and $\FIM_{\gi\gi}$}

\vspace{0.1cm}
Similarly to case \emph{i)}, $\FIM_{\gr\gr}:=\frac2{\noisevar}\Re\{\J_{\bm{gg}}\}$, where
\begin{equation*}
	\begin{aligned}
		[\J_{\g\g}]_{p_1,p_2} :=& \pd{\r^\herm(\boldtheta)}{\Re\{g_{p_1}\}}\pd{\r(\boldtheta)}{\Re\{g_{p_2}\}} \\
		=&\ \s^\herm \TAU^\herm (\e_{p_1}\e_{p_1}^\trans\otimes\I_Q) \NU^\herm \NU(\e_{p_2}\e_{p_2}^\trans\otimes\I_Q)\TAU\s \\
		=&\ Q\Tr\Big[\Big(\NU^\herm\NU\Bodot\TAU^*\Rs^\trans\TAU^\trans\Big)_{\{p_1,p_2\}}\Big].
	\end{aligned}
\end{equation*}
Collectively, $\J_{\g\g}$ can be written as \eqref{eqn:J_g}. Observe that
\begin{equation*}
	\begin{aligned}
		[\FIM_{\gi\gi}]_{p_1,p_2} &= \frac{2}{\noisevar}\Re\Big\{(-\jj)\pd{\r^\herm(\boldtheta)}{\Re\{g_{p_1}\}}(\jj)\pd{\r(\boldtheta)}{\Re\{g_{p_2}\}}\Big\} \\
		&= \frac{2}{\noisevar}\Re\Big\{\pd{\r^\herm(\boldtheta)}{\Re\{g_{p_1}\}}\pd{\r(\boldtheta)}{\Re\{g_{p_2}\}}\Big\};\\
		[\FIM_{\gr\gi}]_{p_1,p_2} &= \frac{2}{\noisevar}\Re\Big\{\pd{\r^\herm(\boldtheta)}{\Re\{g_{p_1}\}}(\jj)\pd{\r(\boldtheta)}{\Re\{g_{p_2}\}}\Big\} \\
		&= -\frac{2}{\noisevar}\Im\Big\{\pd{\r^\herm(\boldtheta)}{\Re\{g_{p_1}\}}\pd{\r(\boldtheta)}{\Re\{g_{p_2}\}}\Big\},\\
	\end{aligned}
\end{equation*}
which shows that $\FIM_{\gi\gi}$ and $\FIM_{\gr\gi}$ are only related to $\J_{\bm{gg}}$, i.e., $\FIM_{\gi\gi}=\frac2\noisevar\Re\{\J_{\bm{gg}}\}$; $\FIM_{\gr\gi}=-\frac2\noisevar\Im\{\J_{\bm{gg}}\}$.

\vspace{0.3cm}
\noindent \emph{iii) Calculation of FIM blocks $\FIM_{\boldnu\gr}$, $\FIM_{\boldnu\gi}$, $\FIM_{\boldtau\gi}$, and $\FIM_{\boldtau\gi}$}

Similarly to case \emph{i)}, $\FIM_{\boldnu\gr}:=\frac2{\noisevar}\Re\{\J_{\boldnu\g}\}$, where
\begin{equation*}
	\begin{aligned}
		[\J_{\boldnu\g}]_{p_1,p_2} :=&\ \pd{\r^\herm(\boldtheta)}{\nu_{p_1}}\pd{\r(\boldtheta)}{\Re\{g_{p_2}\}} \\
		=&\ \s^\herm \TAU^\herm \G^* (\e_{p_1}\e_{p_1}^\trans\otimes\I_Q) \dot{\NU}^\herm \NU(\e_{p_2}\e_{p_2}^\trans\otimes\I_Q)\TAU\s \\
		=&\ Q\Tr\Big[\Big(\dot{\NU}^\herm\NU\Bodot\G^*\TAU^*\Rs^\trans\TAU^\trans\Big)_{\{p_1,p_2\}}\Big],
	\end{aligned}
\end{equation*}
hence, $\J_{\boldnu\g}$ is given in \eqref{eqn:J_nug};
\begin{equation*}
	\begin{aligned}
		\FIM_{\boldnu\gi}&=\frac{2}{\noisevar}\Re\Big\{\pd{\r^\herm(\boldtheta)}{\boldnu}\pd{\r(\boldtheta)}{\gi^\trans}\Big\} =\frac{2}{\noisevar}\Re\Big\{(\jj)\pd{\r^\herm(\boldtheta)}{\boldnu}\pd{\r(\boldtheta)}{\gr^\trans}\Big\}\\
		&= -\frac2{\noisevar} \Im\{\J_{\boldnu\g}\}.
	\end{aligned}
\end{equation*}
Likewise, $\FIM_{\boldtau\gr}$ and $\FIM_{\boldtau\gi}$ can be calculated in a similar way, and \eqref{eqn:J_taug} can be obtained.

Finally, the proof is completed. \qed

\section{}
\label{appendix:FIM_element_wise}

Expanding the expressions in \eqref{eqn:J} using the definitions of $\NU$, $\G$, $\TAU$ and the property of the $\BTr(\cdot)$ operator, the sub-blocks of the FIM can be expressed in an element-wise manner, given in \eqref{eqn:J_element_wise}, where $1\leq p_1,p_2\leq P$. 

\begin{figure*}[!t]
	\begin{equation} \label{eqn:J_element_wise}
		\begin{aligned}
			[\J_{\boldtau\boldtau}]_{p_1,p_2} &= Qg_{p_1}^*g_{p_2}\Tr\big[\V^\herm(\nu_{p_1})\V(\nu_{p_2})\otimes\dot{\T}_{\mathsf{ps}}(\tau_{p_2})\dot{\T}_{\mathsf{ps}}^*(\tau_{p_1})] \Rs^\trans\big]\\
			[\J_{\boldtau\boldnu}]_{p_1,p_2} &= Qg_{p_1}^*g_{p_2}\Tr\big[\V^\herm(\nu_{p_1})\dot{\V}(\nu_{p_2})\otimes\T_{\mathsf{ps}}(\tau_{p_2})\dot{\T}_{\mathsf{ps}}^*(\tau_{p_1})] \Rs^\trans\big]\\
			[\J_{\boldnu\boldnu}]_{p_1,p_2} &=  Qg_{p_1}^*g_{p_2}\Tr\big[\dot{\V}^\herm(\nu_{p_1})\dot{\V}(\nu_{p_2})\otimes\T_{\mathsf{ps}}(\tau_{p_2})\T_{\mathsf{ps}}^*(\tau_{p_1})] \Rs^\trans\big]\\			
			[\J_{\boldtau\g}]_{p_1,p_2} &= Qg_{p_1}^*\Tr\big[\V^\herm(\nu_{p_1})\V(\nu_{p_2})\otimes\T_{\mathsf{ps}}(\tau_{p_2})\dot{\T}_{\mathsf{ps}}^*(\tau_{p_1})] \Rs^\trans\big] \\
			[\J_{\boldnu\g}]_{p_1,p_2} &= Qg_{p_1}^*\Tr\big[\dot{\V}^\herm(\nu_{p_1})\V(\nu_{p_2})\otimes\T_{\mathsf{ps}}(\tau_{p_2})\T_{\mathsf{ps}}^*(\tau_{p_1})] \Rs^\trans\big] \\
			[\J_{\g\g}]_{p_1,p_2} &= Q\Tr\big[\V^\herm(\nu_{p_1})\V(\nu_{p_2})\otimes\T_{\mathsf{ps}}(\tau_{p_2})\T_{\mathsf{ps}}^*(\tau_{p_1}) \Rs^\trans\big] \\
		\end{aligned}
	\end{equation}
	\rule{\linewidth}{0.2pt}
\end{figure*}

\section{Proof of Theorem \ref{theorem:single_target_general}}
\label{appendix:single_target_general}

\emph{Proof:} For the single-target case, each FIM sub-block in \eqref{eqn:J} is a scalar. 
Taking $\mathscr{J}_{\tau\tau}$ in \eqref{eqn:J_element_wise} as an example (here, $p_1=p_2=1$):
\begin{align*}
		\mathscr{J}_{\tau\tau} &= Qg^*g
		\Tr\Big[(\V^\herm\V\otimes\dot{\T}_{\mathsf{ps}}\dot{\T}_{\mathsf{ps}}^*) \Rs^\trans\Big] \\
		&=  Q|g|^2\Tr\Big[ \Rs(\I_N\otimes\dot{\T}_{\mathsf{ps}}^*\dot{\T}_{\mathsf{ps}})\Big].
\end{align*}	
Denote $\Breve{\T}_{\mathsf{ps}}=\I_N\otimes \dot{\T}_{\mathsf{ps}}^*\dot{\T}_{\mathsf{ps}}$, which is diagonal. It holds that
\begin{equation*}
	\begin{aligned}
		\Tr[\Rs\Breve{\T}_{\mathsf{ps}}] &= \sum_{i\in\mathcal{Q}} [\Rs\Breve{\T}_{\mathsf{ps}}]_{i,i} = \sum_{i\in\mathcal{Q}} \sum_{j\in\mathcal{Q}} [\Rs]_{i,j}[\Breve{\T}_{\mathsf{ps}}]_{j,i} \\
		&= \sum_{i\in\mathcal{Q}}[\Rs]_{i,i}[\Breve{\T}_{\mathsf{ps}}]_{i,i}+ \sum_{i\in\mathcal{Q}} \sum_{\substack{j\in\mathcal{Q}\\j\neq i}} [\Rs]_{i,j}\underbrace{[\Breve{\T}_{\mathsf{ps}}]_{j,i}}_{=0}\\
		&= \frac1Q \sum_{i\in\mathcal{Q}} [\Breve{\T}_{\mathsf{ps}}]_{i,i} |[\s]_i|^2.
	\end{aligned}
\end{equation*}
Hence, $\mathscr{J}_{\tau\tau}$ can be further written as \eqref{eqn:J_tau_scalar}. Likewise,
\begin{equation*}
	\begin{aligned}
		\mathscr{J}_{\tau\nu} &= Q|g|^2\Tr\Big[\Rs (\V^*\dot{\V}\otimes\dot{\T}_{\mathsf{ps}}^*\T_{\mathsf{ps}})\Big]; \\
		\mathscr{J}_{\nu\nu} &= Q|g|^2\Tr\Big[\Rs (\dot{\V}^*\dot{\V}\otimes\T_{\mathsf{ps}}^*\T_{\mathsf{ps}}) \Big]; \\
		\mathscr{J}_{\tau g} &= Qg^*\Tr\Big[\Rs (\I_N\otimes\dot{\T}_{\mathsf{ps}}^*\T_{\mathsf{ps}}) \Big]; \\
		\mathscr{J}_{\nu g} &= Qg^*\Tr\Big[\Rs (\dot{\V}^*\V\otimes\T_{\mathsf{ps}}^*\T_{\mathsf{ps}}) \Big]; \\
		\mathscr{J}_{gg} &= Q\Tr\Big[\Rs(\I_N\otimes\T_{\mathsf{ps}}^*\T_{\mathsf{ps}}) \Big], \\
	\end{aligned}
\end{equation*}
yielding \eqref{eqn:J_taunu_scalar}-\eqref{eqn:J_g_scalar} respectively. From the block structure in Theorem \ref{theorem:FIM_structure}, the equivalent FIM (see details in \cite{wang2025device}) for $\tau$ and $\nu$ satisfies
\begin{equation*}
	\begin{aligned}
		\frac\noisevar2\CRB_{\{\tau,\nu\}}^{-1} &= \begin{bmatrix}
			\mathscr{J}_{\tau\tau} & \mathscr{J}_{\tau\nu}\\
			\mathscr{J}_{\tau\nu} & \mathscr{J}_{\nu\nu}
		\end{bmatrix} -
		\begin{bmatrix}
			\Re\{\mathscr{J}_{\tau g}\} & -\Im\{\mathscr{J}_{\tau g}\} \\
			\Re\{\mathscr{J}_{\nu g}\} & -\Im\{\mathscr{J}_{\nu g}\}
		\end{bmatrix}\times\\
		&\begin{bmatrix}
			\mathscr{J}_{gg} & 0 \\
			0 & \mathscr{J}_{gg}
		\end{bmatrix}^{-1}
		\begin{bmatrix}
			\Re\{\mathscr{J}_{\tau g}\} & -\Im\{\mathscr{J}_{\tau g}\} \\
			\Re\{\mathscr{J}_{\nu g}\} & -\Im\{\mathscr{J}_{\nu g}\}
		\end{bmatrix}^\trans\\
	\end{aligned}
\end{equation*}
After certain matrix algebras, \eqref{eqn:EFIM_single_target} is obtained. \qed

\section{Proof of Proposition \ref{proposition:CRB_single_target}}
\label{appendix:CRB_single_target}

\emph{Proof:} For the ideal sinc pulse shaping, $ \Upsilon_{\mathsf{sinc}}(m)=1$. Under isotropic transmission, $|[\s]_{q'}|^2=\pt$ for $q'\in\mathcal{Q}$. From \eqref{eqn:J_tau_scalar}, it can be calculated that
\begin{equation*}
	\begin{aligned}
		\mathscr{J}_{\tau\tau}&=\pt|g|^2\sum_{q'\in\mathcal{Q}}\Big(\frac{2\pi \eta_m}{\Ts}\Big)^2= \pt|g|^2\frac{4\pi^2N}{\Ts^2} \sum_{m=0}^{M-1}\eta_m^2.\\
	\end{aligned}
\end{equation*}
Similarly,
\begin{equation*}
	\begin{aligned}
		\mathscr{J}_{\tau\nu} &= -\pt|g|^2  \frac{2\pi^2\Tsym N(N-1)}{\Ts}\sum_{m=0}^{M-1}\eta_m;\\
		\mathscr{J}_{\nu\nu} &= \pt|g|^2 \frac{2\pi^2Q(N-1)(2N-1)\Tsym^2}{3};\\
		\mathscr{J}_{\tau g} &= \jj \pt g^* \frac{2\pi N}{\Ts}\sum_{m=0}^{M-1}\eta_m; \\
		\mathscr{J}_{\nu g} &= -\jj \pt g^* \pi Q(N-1)\Tsym; \\
		\mathscr{J}_{gg} &= Q\pt.
	\end{aligned}
\end{equation*}

From \eqref{eqn:EFIM_single_target}, it can be further obtained that
\begin{align*}
	& \mathscr{J}_{\tau\tau}-\frac{|\mathscr{J}_{\tau g}|^2}{\|\s\|^2} = \pt|g|^2\frac{4\pi^2N}{\Ts^2}\Big[\sum_{m=0}^{M-1}\eta_m^2-\frac1M\Big(\sum_{m=0}^{M-1}\eta_m\Big)^2\Big];\\
	& \mathscr{J}_{\tau\nu}+\frac{|\mathscr{J}_{\tau g}\mathscr{J}_{\nu g}|}{\|\s\|^2}
	=0;\\
	& \mathscr{J}_{\nu\nu}-\frac{|\mathscr{J}_{\nu g}|^2}{\|\s\|^2} =\pt|g|^2\frac{\pi^2Q(N^2-1)\Tsym^2}{3}.
\end{align*}

The off-diagonal elements in the equivalent FIM are zero, suggesting that estimation of $\tau$ and $\nu$ is actually decoupled. It is further calculated that
\begin{equation*}
	\begin{aligned}
		\sum_{m=0}^{M-1}\eta_m&=\begin{cases}
			-1/2,& M\text{ is even}\\
			-1,&M\text{ is odd}
		\end{cases};\\
		\sum_{m=0}^{M-1}\eta_m^2&=\begin{cases}
			\frac{M^2+2}{12M},& M\text{ is even}\\
			\frac{M^2+11}{12M},&M\text{ is odd}
		\end{cases}.\\
	\end{aligned}
\end{equation*}
Therefore, $\sum_{m=0}^{M-1}\eta_m^2-\frac1M\big(\sum_{m=0}^{M-1}\eta_m\big)^2= \frac{M^2-1}{12M}$ for any value of $M$. The corresponding CRB is therefore evaluated as $	\varepsilon(u) = \big[\frac2\noisevar\big(	\mathscr{J}_{uu}-\frac{|\mathscr{J}_{ug}|^2}{\|\s\|^2}\big)\big]^{-1}$, $u\in\{\tau,\nu\}$, yielding \eqref{eqn:CRB_single_target}. \qed

\section{Proof of Proposition \ref{Proposition:optimization_transform}}
\label{appendix:optimization_transform}

\emph{Proof:} The original optimization objective in \eqref{eqn:CRB_optimization_original} admits the following equivalent transformation
\begin{equation*}
	\begin{aligned}
		&\ \alpha \Tr(\CRB_{\tilde\boldtau}) + (1-\alpha)\Tr(\CRB_{\tilde\boldnu}) \\
		=&\ \Tr(\Mdiag\{\alpha\I_P,(1-\alpha)\I_P,\mathbf{0}_{2P}\}\CRB_{\tilde\boldtheta}) \\
		=&\ \Tr(\Mdiag\{\alpha\I_P,(1-\alpha)\I_P,\mathbf{0}_{2P}\}\D\CRB_{\boldtheta}\D) \\
		=&\ \Tr(\underbrace{\D\Mdiag\{\alpha\I_P,(1-\alpha)\I_P,\mathbf{0}_{2P}\}\D}_{:=\A(\alpha)}\CRB_{\boldtheta})\\
		=&\ \Tr[\A(\alpha)\FIM_{\boldtheta\boldtheta}^{-1}].
	\end{aligned}
\end{equation*}
From Proposition \ref{proposition:dependence_diagonal}, the FIM $\FIM_{\boldtheta\boldtheta}$ depends upon the diagonal terms of $\Rs$. Denote $\boldrho=[\varrho_1,\varrho_2,\cdots,\varrho_Q]^\trans$ as the collection of diagonal terms of $\Rs$, we may temporarily neglect the rank-$1$ constraint\footnote{This is because, as long as the diagonal terms of $\Rs$ are identical, the resulting FIM unchanges regardless of its rank (see Proposition \ref{proposition:dependence_diagonal}).} and seek the solution to \eqref{eqn:CRB_optimization_original} taking the form of
\begin{equation*}
	\Rs^\circ = {\Mdiag}\{\boldrho\} = \sum_{q'\in\mathcal{Q}} \e_{q'}\e_{q'}^\trans \varrho_{q'}.
\end{equation*}
From the Slepian-Bangs formula,
\begin{equation*}
	\begin{aligned}
		&[\FIM_{\boldtheta\boldtheta}]_{i,j} = \frac{2}{\noisevar} \Re\left\{ \pd{\r^\herm(\boldtheta)}{[\boldtheta]_i}\pd{\r(\boldtheta)}{[\boldtheta]_j} \right\} \\
		=&\ \frac{2Q}{\noisevar} \Re\left\{\Tr\left\{\Rs^\circ \pd{(\NU\G\TAU)^\herm}{[\boldtheta]_i}\pd{(\NU\G\TAU)}{[\boldtheta]_j} \right\} \right\}\\
		=&\ \sum_{q'\in\mathcal{Q}} \varrho_{q'} \underbrace{\frac{2Q}{\noisevar} \Re\left\{\Tr\left[\e_{q'}\e_{q'}^\trans \pd{(\NU\G\TAU)^\herm}{[\boldtheta]_i}\pd{(\NU\G\TAU)}{[\boldtheta]_j} \right] \right\}}_{[\FIM_{\boldtheta\boldtheta}]_{i,j}\text{ evaluated at }\Rs=\e_{q'}\e_{q'}^\trans}.
	\end{aligned}
\end{equation*}
Hence, the overall $\FIM_{\boldtheta\boldtheta}$ can be decomposed as \eqref{eqn:FIM_decompose}; the original optimization problem \eqref{eqn:CRB_optimization_original} can be transformed into \eqref{eqn:CRB_optimization_reduced}. \qed

\balance
\bibliographystyle{IEEEtran}
\bibliography{ref}

@ARTICLE{wang_2026_TIT,
	author={Wang, Zi-Jie and Wang, Xudong and Caire, Giuseppe},
	journal={IEEE Trans. Inf. Theory}, 
	title={On {P}areto-Optimal Estimation-Information Performance Limits of {MIMO} Integrated Sensing and Communications Systems}, 
	year={2026},
	note={{early access, DOI: 10.1109/TIT.2026.3717854}}}

@ARTICLE{OFDM_AMT_TWC,
	author={Wang, Shuhan and Tang, Aimin and Wang, Xudong and Qu, Wenze},
	journal={IEEE Trans. Wireless Commun.}, 
	title={Sensing-Assisted Channel Estimation for Bistatic {OFDM ISAC} Systems: Framework, Algorithm, and Analysis}, 
	year={2026},
	volume={25},
	number={},
	pages={7842-7857},
	doi={10.1109/TWC.2025.3633948}}

@ARTICLE{TVT_BA,
	author={Wang, Zi-Jie and Tang, Aimin and Wang, Xudong},
	journal={IEEE Trans. Veh. Technol.}, 
	title={Theoretical Performance Study for {mmWave/THz ISAC} Systems with Sensing-Based Beam Alignment}, 
	year={2026},
	volume={75},
	number={8},
	pages={17562-17575},
	doi={10.1109/TVT.2026.3679354}}

@ARTICLE{Multicarrier_survey,
	author={Zhang, Xingyao and Yin, Haoran and Tang, Yanqun and Ge, Yao and Zeng, Yong and Wen, Miaowen and Liu, Zilong and Guan, Yong Liang and Arslan, Hüseyin and Caire, Giuseppe},
	journal={IEEE Commun. Surveys Tuts.}, 
	title={A Unified Multicarrier Waveform Framework for Next-Generation Wireless Networks: Principles, Performance, and Challenges}, 
	year={2026},
	volume={28},
	number={},
	pages={5416-5455},
	doi={10.1109/COMST.2026.3672602}}

@ARTICLE{Delay_Doppler_processing_numerical2,
  author={Kenney, John B.},
  journal={Proc. IEEE}, 
  title={Dedicated Short-Range Communications {(DSRC)} Standards in the{ United States}}, 
  year={2011},
  volume={99},
  number={7},
  pages={1162-1182},
  doi={10.1109/JPROC.2011.2132790}}

@ARTICLE{ODDM+FMCW,
	author={Huang, Kehan and Shafie, Akram and Qiu, Min and Aboutanios, Elias and Yuan, Jinhong},
	journal={IEEE Trans. Wireless Commun.}, 
	title={A Novel {ISAC} Waveform Based on Orthogonal Delay-{Doppler} Division Multiplexing With {FMCW}}, 
	year={2026},
	volume={25},
	number={},
	pages={11688-11705},
	doi={10.1109/TWC.2026.3661286}}

@INPROCEEDINGS{Delay_Doppler_processing_numerical1,
  author={Nguyen, Duy H. N. and Heath, Robert W.},
  booktitle={IEEE ICASSP}, 
  title={Delay and {Doppler} processing for multi-target detection with {IEEE} 802.11 {OFDM} signaling}, 
  year={2017},
  volume={},
  number={},
  pages={3414-3418},
  doi={10.1109/ICASSP.2017.7952790}}

@ARTICLE{FanLiu_iceberg_TSP,
  author={Liu, Fan and Xiong, Yifeng and Lu, Shihang and Li, Shuangyang and Yuan, Weijie and Masouros, Christos and Jin, Shi and Caire, Giuseppe},
  journal={IEEE Trans. Signal Process.}, 
  title={Uncovering the Iceberg in the Sea: Fundamentals of Pulse Shaping and Modulation Design for Random {ISAC} Signals}, 
  year={2025},
  volume={73},
  number={},
  pages={2511-2526},
  doi={10.1109/TSP.2025.3580596}}

@ARTICLE{DD_waveform,
  author={Yuan, Weijie and Zhou, Lin and Dehkordi, Saeid K. and Li, Shuangyang and Fan, Pingzhi and Caire, Giuseppe and Poor, H. Vincent},
  journal={IEEE Wireless Commun.}, 
  title={From {OTFS} to {DD-ISAC}: Integrating Sensing and Communications in the Delay {Doppler} Domain}, 
  year={2024},
  volume={31},
  number={6},
  pages={152-160},
  doi={10.1109/MWC.018.2300607}}

@ARTICLE{Overview_waveform_OTFS,
  author={Wei, Zhiqiang and Yuan, Weijie and Li, Shuangyang and Yuan, Jinhong and Bharatula, Ganesh and Hadani, Ronny and Hanzo, Lajos},
  journal={IEEE Wireless Commun.}, 
  title={Orthogonal Time-Frequency Space Modulation: A Promising Next-Generation Waveform}, 
  year={2021},
  volume={28},
  number={4},
  pages={136-144},
  doi={10.1109/MWC.001.2000408}}

@ARTICLE{Overview_waveform_AFDM,
	author={Rou, Hyeon Seok and Ranasinghe, Kuranage Roche Rayan and Savaux, Vincent and Abreu, Giuseppe Thadeu Freitas de and David González, G. and Masouros, Christos},
	journal={IEEE Commun. Stand. Mag.}, 
 	title={Affine Frequency Division Multiplexing ({AFDM}) for {6G}: Properties, Features, and Challenges}, 
	year={2026},
	volume={10},
	number={2},
	pages={216-225},
	doi={10.1109/MCOMSTD.2025.3643183}}

@ARTICLE{AFDM_communication,
  author={Bemani, Ali and Ksairi, Nassar and Kountouris, Marios},
  journal={IEEE Trans. Wireless Commun.}, 
  title={Affine Frequency Division Multiplexing for Next Generation Wireless Communications}, 
  year={2023},
  volume={22},
  number={11},
  pages={8214-8229},
  doi={10.1109/TWC.2023.3260906}}

@ARTICLE{OTFS_ISAC,
  author={Yuan, Weijie and Wei, Zhiqiang and Li, Shuangyang and Yuan, Jinhong and Ng, Derrick Wing Kwan},
  journal={IEEE J. Sel. Topics Signal Process.}, 
  title={Integrated Sensing and Communication-Assisted Orthogonal Time Frequency Space Transmission for Vehicular Networks}, 
  year={2021},
  volume={15},
  number={6},
  pages={1515-1528},
  doi={10.1109/JSTSP.2021.3117404}}

@INPROCEEDINGS{OTFS_CRB_uplink,
  author={Nie, Mingcheng and Li, Shuangyang and Mishra, Deepak},
  booktitle={IEEE ICC Workshops}, 
  title={Improving Channel Estimation Performance for Uplink {OTFS} Transmissions: Pilot Design based on A Posteriori {Cram\'{e}r-Rao} Bound}, 
  year={2023},
  pages={301-306}}

@ARTICLE{OFDM+OTFS2,
  author={Gaudio, Lorenzo and Kobayashi, Mari and Caire, Giuseppe and Colavolpe, Giulio},
  journal={IEEE Trans. Wireless Commun.}, 
  title={On the Effectiveness of {OTFS} for Joint Radar Parameter Estimation and Communication}, 
  year={2020},
  volume={19},
  number={9},
  pages={5951-5965},
  doi={10.1109/TWC.2020.2998583}}

@ARTICLE{OTFS_ISAC_MIMO,
  author={Keskin, Musa Furkan and Marcus, Carina and Eriksson, Olof and Alvarado, Alex and Widmer, Joerg and Wymeersch, Henk},
  journal={IEEE Trans. Wireless Commun.}, 
  title={Integrated Sensing and Communications With {MIMO-OTFS}: {ISI/ICI} Exploitation and Delay-{Doppler} Multiplexing}, 
  year={2024},
  volume={23},
  number={8},
  pages={10229-10246},
  doi={10.1109/TWC.2024.3370501}}

@article{entropy_OTFS,
  title={Cross-Domain {OTFS} Detection via Delay-{Doppler} Decoupling: Reduced-Complexity Design and Performance Analysis},
  author={Liu, Mengmeng and Li, Shuangyang and Bai, Baoming and Caire, Giuseppe},
  journal={Entropy},
  volume={27},
  number={10},
  pages={1062},
  year={2025},
  publisher={MDPI}
}

@INPROCEEDINGS{AFDM_ISAC,
  author={Bao, Hongjie and Zhuang, Hongcheng and Wang, Zhaocheng and Pang, Gaokun},
  booktitle={IEEE PIMRC}, 
  title={Performance Trade-off between Communication and Sensing Based on {AFDM} Parameter Adjustment}, 
  year={2024},
  volume={},
  number={},
  pages={1-6},
  doi={10.1109/PIMRC59610.2024.10817290}}

@misc{cao2025agileaffinefrequencydivision,
      title={Agile Affine Frequency Division Multiplexing}, 
      author={Yewen Cao and Yulin Shao},
      year={2025},
      eprint={2512.14424},
      archivePrefix={arXiv},
      primaryClass={cs.IT},
      url={https://arxiv.org/abs/2512.14424}, 
}

@ARTICLE{AFDM_ISAC_MIMO,
  author={Luo, Yirui and Guan, Yong Liang and Ge, Yao and González G, David and Yuen, Chau},
  journal={IEEE Internet Things J.}, 
  title={A Novel Angle-Delay-{Doppler} Estimation Scheme for {AFDM-ISAC} System in Mixed Near-Field and Far-Field Scenarios}, 
  year={2025},
  volume={12},
  number={13},
  pages={22669-22682},
  doi={10.1109/JIOT.2025.3559924}}

@ARTICLE{ODDM_overall,
  author={Lin, Hai and Yuan, Jinhong},
  journal={IEEE Trans. Wireless Commun.}, 
  title={Orthogonal Delay-{Doppler} Division Multiplexing Modulation}, 
  year={2022},
  volume={21},
  number={12},
  pages={11024-11037},
  doi={10.1109/TWC.2022.3188776}}

@ARTICLE{ODDM_ISAC,
  author={Li, Meilin and Han, Chong and Jin, Shi},
  journal={IEEE J. Sel. Areas Commun.}, 
  title={Hybrid Beamforming with Orthogonal Delay-{Doppler} Division Multiplexing Modulation for {Terahertz} Sensing and Communication}, 
  year={2026},
  volume={44},
  number={},
  pages={746-759},
  doi={10.1109/JSAC.2025.3610799}}

@ARTICLE{OCDM_overall,
  author={Omar, Muhammad Shahmeer and Ma, Xiaoli},
  journal={IEEE Trans. Wireless Commun.}, 
  title={Performance Analysis of {OCDM} for Wireless Communications}, 
  year={2021},
  volume={20},
  number={7},
  pages={4032-4043},
  doi={10.1109/TWC.2021.3055070}}

@ARTICLE{OCDM_original,
  author={Ouyang, Xing and Zhao, Jian},
  journal={IEEE Trans. Commun.}, 
  title={Orthogonal Chirp Division Multiplexing}, 
  year={2016},
  volume={64},
  number={9},
  pages={3946-3957},
  doi={10.1109/TCOMM.2016.2594792}}

@ARTICLE{xia2025jointlocationvelocityestimation,
	author={Xia, Guoqing and Xiao, Pei and Luo, Qu and Ji, Bing and Zhang, Yue and Zhou, Huiyu},
	journal={IEEE Trans. Commun.}, 
	title={Joint Location and Velocity Estimation and Fundamental {CRLB} Analysis for Cell-Free {MIMO-ISAC}}, 
	year={2026},
	volume={74},
	number={},
	pages={7357-7374}}

@ARTICLE{Miller1978MCRB,
	author={Miller, R. and Chow Chang},
	journal={IEEE Trans. Inf. Theory},
	title={A modified {Cram\'{e}r-Rao} bound and its applications (Corresp.)},
	year={1978},
	volume={24},
	number={3},
	pages={398-400},
	doi={10.1109/TIT.1978.1055879}}

@ARTICLE{FanLiu6G,
  author={Liu, Fan and Cui, Yuanhao and Masouros, Christos and Xu, Jie and Han, Tony Xiao and Eldar, Yonina C. and Buzzi, Stefano},
  journal={IEEE J. Sel. Areas Commun.}, 
  title={Integrated Sensing and Communications: Toward Dual-Functional Wireless Networks for {6G} and Beyond}, 
  year={2022},
  volume={40},
  number={6},
  pages={1728-1767}}

@ARTICLE{multi-tone,
	author={Rife, D. C. and Boorstyn, R. R.},
	journal={The Bell System Technical Journal},
	title={Multiple tone parameter estimation from discrete-time observations},
	year={1976},
	volume={55},
	number={9},
	pages={1389-1410},
	doi={10.1002/j.1538-7305.1976.tb02941.x}
}

@ARTICLE{TWC_Hua,
  author={Hua, Haocheng and Han, Tony Xiao and Xu, Jie},
  journal={IEEE Trans. Wireless Commun.}, 
  title={{MIMO} Integrated Sensing and Communication: {CRB}-Rate Tradeoff}, 
  year={2024},
  volume={23},
  number={4},
  pages={2839-2854},
  doi={10.1109/TWC.2023.3303326}}

@ARTICLE{MIMO_radar_signal_design_TSP,
  author={Stoica, Petre and Li, Jian and Xie, Yao},
  journal={IEEE Trans. Signal Process.}, 
  title={On Probing Signal Design For {MIMO} Radar}, 
  year={2007},
  volume={55},
  number={8},
  pages={4151-4161}}

@phdthesis{braun2014ofdm,
  title={{OFDM} radar algorithms in mobile communication networks},
  author={Braun, Klaus Martin},
  year={2014},
  school={Commun. Eng. Lab, Karlsruher Institut für Technologie (KIT), Karlsruhe, Germany}
}

@book{Kay1993Fundamentals,
	author = {Kay, Steven M.},
	title = {Fundamentals of Statistical Signal Processing: Estimation Theory},
	year = {1998},
	publisher = {Prentice-Hall},
	address = {Englewood Cliffs, NJ, USA}
}

@ARTICLE{ISAC_OTFS_TWC,
	author={Keskin, Musa Furkan and Marcus, Carina and Eriksson, Olof and Alvarado, Alex and Widmer, Joerg and Wymeersch, Henk},
	journal={IEEE Trans. Wireless Commun.},
	title={Integrated Sensing and Communications With {MIMO-OTFS: ISI/ICI} Exploitation and Delay-{Doppler} Multiplexing},
	year={2024},
	volume={23},
	number={8},
	pages={10229-10246},
}

@ARTICLE{FanLiu_TIT,
	author={Liu, Fan and Zhang, Ying and Xiong, Yifeng and Li, Shuangyang and Yuan, Weijie and Gao, Feifei and Jin, Shi and Caire, Giuseppe},
	journal={IEEE Trans. Inf. Theory},
	title={{CP-OFDM} Achieves the Lowest Average Ranging Sidelobe Under {QAM/PSK} Constellations},
	year={2025},
	volume={71},
	number={9},
	pages={6950-6967},
	doi={10.1109/TIT.2025.3591267}}

@ARTICLE{OFDMRadar_IEEE_Proc,
	author={Sturm, Christian and Wiesbeck, Werner},
	journal={Proc. IEEE},
	title={Waveform Design and Signal Processing Aspects for Fusion of Wireless Communications and Radar Sensing},
	year={2011},
	volume={99},
	number={7},
	pages={1236-1259},
	doi={10.1109/JPROC.2011.2131110}}

@article{wang2025device,
	author={Wang, Zi-Jie and Tang, Aimin and Wang, Xudong},
	journal={IEEE J. Sel. Areas Commun.}, 
	title={Device-Free Localization in {ISAC} Networks: Performance Limits and {Fisher} Information-Based Cooperative Localization}, 
	year={2026},
	volume={44},
	number={},
	pages={626-641},
	doi={10.1109/JSAC.2025.3610408}}

@ARTICLE{range_compression_TSP,
	author={Li, Jian and Xu, Luzhou and Stoica, Petre and Forsythe, Keith W. and Bliss, Daniel W.},
	journal={IEEE Trans. Signal Process.},
	title={Range Compression and Waveform Optimization for {MIMO} Radar: A {Cram\'{e}r–Rao bound} Based Study},
	year={2008},
	volume={56},
	number={1},
	pages={218-232},
	doi={10.1109/TSP.2007.901653}}

@book{Complex_Stochastic_Processes,
	author = {Miller, K. S.},
	title = {Complex Stochastic Processes: An Introduction to Theory and Application},
	publisher={Addison-Wesley Publishing Company, Inc.},
	year = {1977},
	address = {Reading, MA, USA}
}

@INPROCEEDINGS{OFDM+OTFS,
	author={Gaudio, Lorenzo and Kobayashi, Mari and Bissinger, Bjorn and Caire, Giuseppe},
	booktitle={IEEE ICC Workshops},
	title={Performance Analysis of Joint Radar and Communication using {OFDM} and {OTFS}},
	year={2019},
	pages={1-6},
	doi={10.1109/ICCW.2019.8757044}}
\begin{IEEEbiography}[{\includegraphics[width=1in,height=1.25in,clip,keepaspectratio]{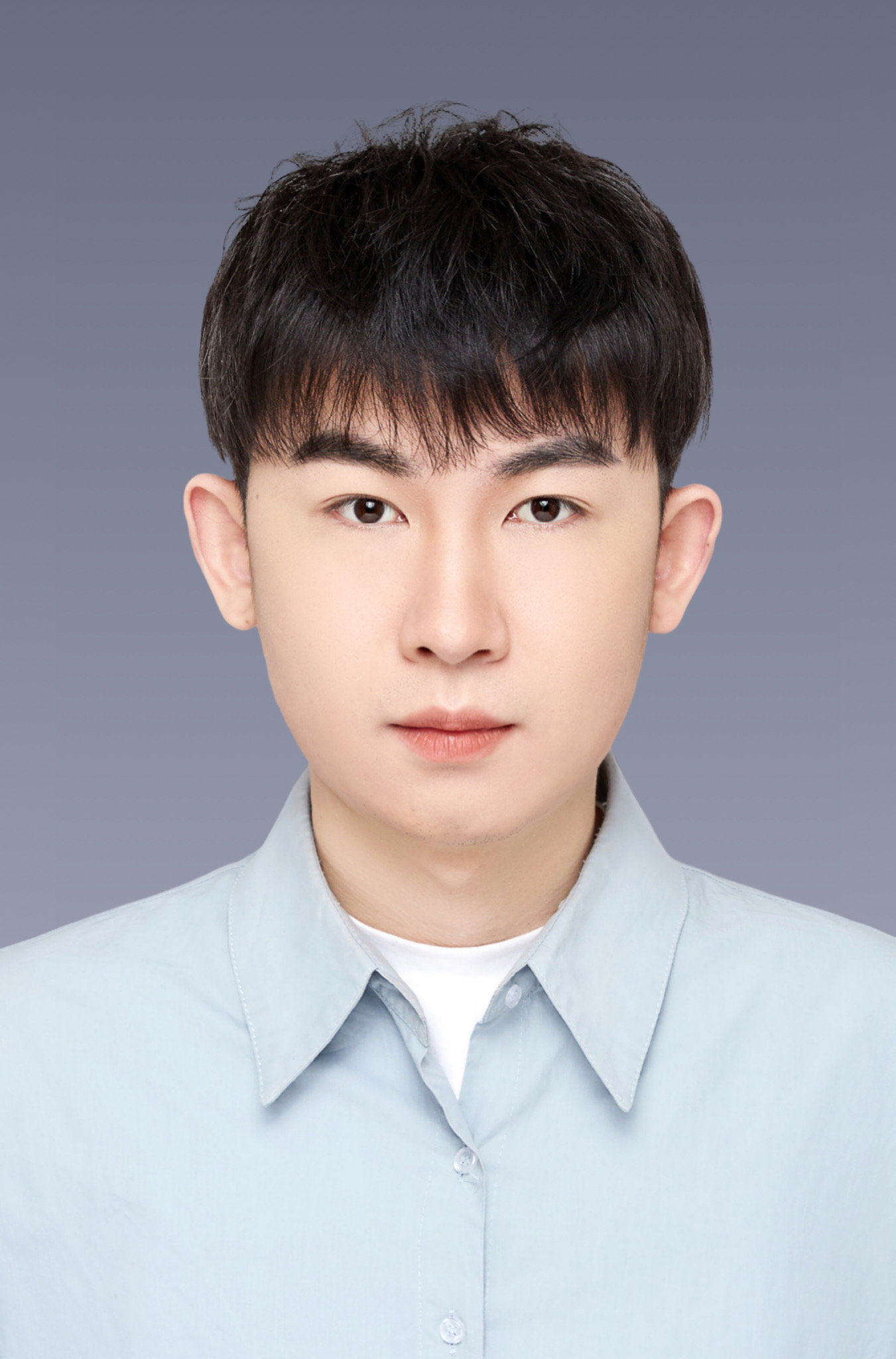}}]{Zi-Jie Wang} 
	(Graduate Student Member, IEEE) received the Bachelor's degree in electronic and computer engineering from Shanghai Jiao Tong University, Shanghai, China, in 2020, and the Master's degree in electronic engineering (signal processing) from Nanyang Technological University, Singapore, in 2022. He is currently pursuing the joint Ph.D. degree in information and communication engineering from Shanghai Jiao Tong University and the National University of Singapore, Singapore. His research interests involve signal processing and information theory, with their applications in wireless sensing, communications, and integrated sensing and communications systems.
\end{IEEEbiography}

\begin{IEEEbiography}[{\includegraphics[width=1in,height=1.25in,clip,keepaspectratio]{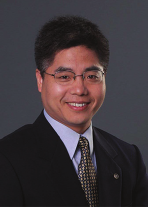}}]{Xudong Wang}
	(Fellow, IEEE) received the PhD degree in electrical and computer engineering from
	the Georgia Institute of Technology in 2003. He is currently a professor with the thrust of Internet of Things and the dean of the College of Future Technology in Hong Kong University of Science and Technology (Guangzhou). He is also an affiliate professor with the Department of Electrical and Computer Engineering, University of Washington.
	He was the John Wu and Jane Sun chair professor with UM-SJTU Joint Institute, Shanghai Jiao Tong
	University. He was a senior research engineer, a senior network architect, and an
	R\&D manager with several companies. His research interests include wireless
	communication networks, distributed machine learning, edge computing, and
	joint communications and sensing. He was the editor of \emph{IEEE Transactions
		on Mobile Computing}, \emph{IEEE Transactions on Vehicular Technology}, \emph{Elsevier
		Ad Hoc Networks}, and \emph{China Communications}. He was the guest editor of
	several international journals. He was the general chair of 2017 IEEE 5G
	Summit in Shanghai and the TPC co-chair of the 32nd International Conference
	on Information Networking. He was the demo co-chair of ACM International
	Symposium on Mobile Ad Hoc Networking and Computing (ACM MOBIHOC
	2006), technical program co-chair of Wireless Internet Conference (WICON)
	2007, and general co-chair of WICON 2008. He was a voting member of IEEE
	802.11 and 802.15 Standard Committees.
\end{IEEEbiography}

\begin{IEEEbiography}[{\includegraphics[width=1in,height=1.25in,clip,keepaspectratio]{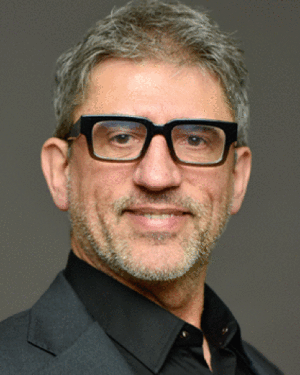}}]{Giuseppe Caire} (S '92 -- M '94 -- SM '03 -- F '05)  was born in Torino in 1965. He received a
B.Sc. in Electrical Engineering  from Politecnico di Torino in 1990,  an M.Sc. in Electrical Engineering from Princeton University in 1992, and a Ph.D. from Politecnico di Torino in 1994.  He has been a post-doctoral research fellow with the European Space Agency (ESTEC, Noordwijk, The Netherlands) in 1994-1995, Assistant Professor in Telecommunications at the Politecnico di Torino, Associate Professor at the University of Parma, Italy,  Professor with the Department of Mobile Communications at the Eurecom Institute,  Sophia-Antipolis, France, a Professor of Electrical Engineering with the Viterbi School of Engineering, University of Southern California, Los Angeles, and he is currently an Alexander von Humboldt Professor with the Faculty of Electrical Engineering and Computer Science at the Technical University of Berlin, Germany. He is a member of the German National Academy of Sciences (Leopoldina) since 2024, and a member of the American National Academy of Engineering (NAE) since 2026. He received the Jack Neubauer Best System Paper Award from the IEEE Vehicular Technology Society in 2003,  the
IEEE Communications Society and Information Theory Society Joint Paper Award in 2004, in 2011, and in 2025, 
the Okawa Research Award in 2006,   the Alexander von Humboldt Professorship in 2014, the Vodafone Innovation Prize in 2015, an ERC Advanced Grant in 2018,  the Leonard G. Abraham Prize for best IEEE JSAC paper in 2019, the IEEE Communications Society Edwin Howard Armstrong Achievement Award in 2020, the 2021 Leibniz Prize  of the German National Science Foundation (DFG), and the  CTTC Technical Achievement Award of the IEEE Communications Society in 2023.  Giuseppe Caire is a Fellow of IEEE since 2005.  He has served in the Board of Governors of the IEEE Information Theory Society from 2004 to 2007, and as officer from 2008 to 2013. He was President of the IEEE Information Theory Society in 2011. 
His main research interests are in the field of communications theory, information theory, channel and source coding
with particular focus on wireless communications.   
\end{IEEEbiography}

\end{document}